\documentclass[12pt,a4paper]{article}
\usepackage[utf8]{inputenc}

\usepackage{caption}

\usepackage[letterpaper,top=2cm,bottom=2cm,left=3cm,right=3cm,marginparwidth=1.75cm]{geometry}

\usepackage{amsmath}
\usepackage{graphicx}
\usepackage[colorlinks=true, allcolors=blue]{hyperref}
\usepackage[T1]{fontenc}
\usepackage{tikz}
\usepackage{natbib}
\usepackage{epigraph}
\usepackage{soul}
\usepackage{mathrsfs}
\usepackage{epsfig}
\usepackage{enumitem}
\usepackage{latexsym}
\usepackage{makeidx}
\usepackage{fancyhdr}
\usepackage{xcolor, graphicx}
\usepackage{color}
\usepackage{pdfcolmk}
\usepackage[greek,english]{babel} 
\usepackage[margin=10pt,font=small,labelfont=bf,labelsep=endash]{caption}
\usepackage{amsthm,amssymb,amsmath}
\usepackage{setspace}
\usepackage[colorinlistoftodos]{todonotes}
\usepackage{mathtools}

\DeclarePairedDelimiter\floor{\lfloor}{\rfloor}

\newtheorem{thm}{Theorem}
\newtheorem{prop}{Proposition}

\theoremstyle{definition}
\newtheorem{example}{Example}
\theoremstyle{plain}
\newtheorem{lemma}{Lemma}
\usepackage{bbm}

\title{\textbf{Information Aggregation with Delegation of Votes\footnote{We are thankful to Alessandra Casella, Timothy Feddersen, Navin Kartik, Mark Dean, and Carlo Prato and to seminar participants at Columbia University, TU Berlin and Indian Statistical Institute, New Delhi, Delhi School of Economics,  for their excellent comments and suggestions. Dhillon and Kotsialou acknowledge funding from EPSRC grant number EP/P031811/1. This work supersedes \cite{dhillon2021information} and \cite{dilip2021information} which are subsumed into the current paper. The authors' emails are: amrita.dhillon@kcl.ac.uk (AD), g.m.kotsialou@lse.ac.uk (GK), dilip.ravindran@hu-berlin.de (DR), xefteris.dimitrios@ucy.ac.cy (DX).}} \\
}
\author{Amrita Dhillon \\ King's College, London\\ \and
 Grammateia Kotsialou\\ London School of Economics\\ \and Dilip Ravindran\\ Humboldt University of Berlin\\ \and
Dimitrios Xefteris\\ University of Cyprus}

\begin{document}
\maketitle

\begin{abstract}

 Liquid democracy is a system that combines  aspects of direct democracy and representative democracy by allowing voters to either vote directly themselves, or delegate their votes to others. In this paper we study the information aggregation properties of liquid democracy in a setting with heterogeneously informed truth-seeking voters\textemdash who want the election outcome to match an underlying state of the world\textemdash and partisan voters. We establish that liquid democracy admits equilibria which improve welfare and information aggregation over direct and representative democracy when voters' preferences and information precisions are publicly or privately known. Liquid democracy also admits equilibria which do worse than the other two systems. We discuss features of efficient and inefficient equilibria and provide conditions under which voters can more easily coordinate on the efficient equilibria in liquid democracy than the other two systems.



\end{abstract}
\pagenumbering{arabic} 

\onehalfspacing
\pagestyle{plain}
\setlength{\parskip}{5pt plus 5pt minus 5pt}
\newpage

\section{Introduction}
\label{sec: intro}

Recent progress on distributed ledger technologies and, in particular, in blockchain based aggregation processes has opened up the possibilities of new and improved  ways of voting\footnote{\cite{DBLP:conf/atal/Brill18} discusses emerging ideas of using the internet for new interactive collective decision-making processes.} (see e.g. \citealp{Dhillon20}). {\it Liquid Democracy}, which was first suggested by \cite{Miller1969} and discussed by \cite{Shubik1970},\footnote{Introducing the idea of "Homo Politicus": voters who are well informed about politics despite not being directly involved in it.} is one such electoral rule that combines aspects of direct  and representative democracy. Under liquid democracy, a voter chooses whether to vote themselves or to delegate their vote to another voter (potentially someone more knowledgeable on a particular issue). Voters can choose to vote or delegate differently for every decision.

Several versions of liquid democracy have been proposed in the literature (e.g. where only one round of delegation is allowed or where voters who receive delegated votes may further delegate these votes\footnote{Delegated votes can  transitively `travel' until a voter who casts is reached (usually known as a `guru').  This `travel' route is decided by the delegation rule applied, which is a function assigning the votes (of each voter who delegates) into a delegation `path' ultimately reaching a guru, see e.g. models with preferences over delegates: \citealp{KotsialouR20, ColleyGN20, Colley21, Brilletal22}. }), some of these versions have already been used in settings of applied interest (e.g. by certain political parties in Europe, Argentina and---more recently---also by mutual funds in the US).\footnote{Google Votes is an implementation of a liquid democracy platform, and Pirate parties in Europe have been using this for a while. Also, as mentioned in \citet{KotsialouR20} and  \cite{DBLP:conf/aaai/BloembergenGL19}, the Democracy Earth Foundation and the EU Horizon 2020 project WeGov-Now use similar platforms. Mutual funds in the US that adopted pass-through voting (e.g. BlackRock and Vanguard), give the option to their investors to either let the fund's management vote on their account, or to vote directly themselves (see, e.g., \citealp{malenko2023voting}). We emphasize that enhancing shareholders' voice within organizations by implementing more efficient and democratic management of their rights, such as internal liquid democracy systems where the role of `gurus' (i.e. their proxy advisors) can be accurately and transparently tracked, can be supported by distributed ledger technologies (see e.g. \citealp{DBLP:conf/atal/KotsialouRDMMMP18, RileyKDMMP19, Dhillon20, LafarreVdE21}).}  Figure \ref{fig:arch} illustrates the features of liquid democracy contrasted with direct and representative democracy.

 \begin{center}

\begin{tabular}{c c c c c}
 \textsc{\small{Representative Democracy}} & & \textsc{\small{Direct Democracy}} & &\textsc{\small{Liquid Democracy}}\\
  \includegraphics[width=0.2 \textwidth]{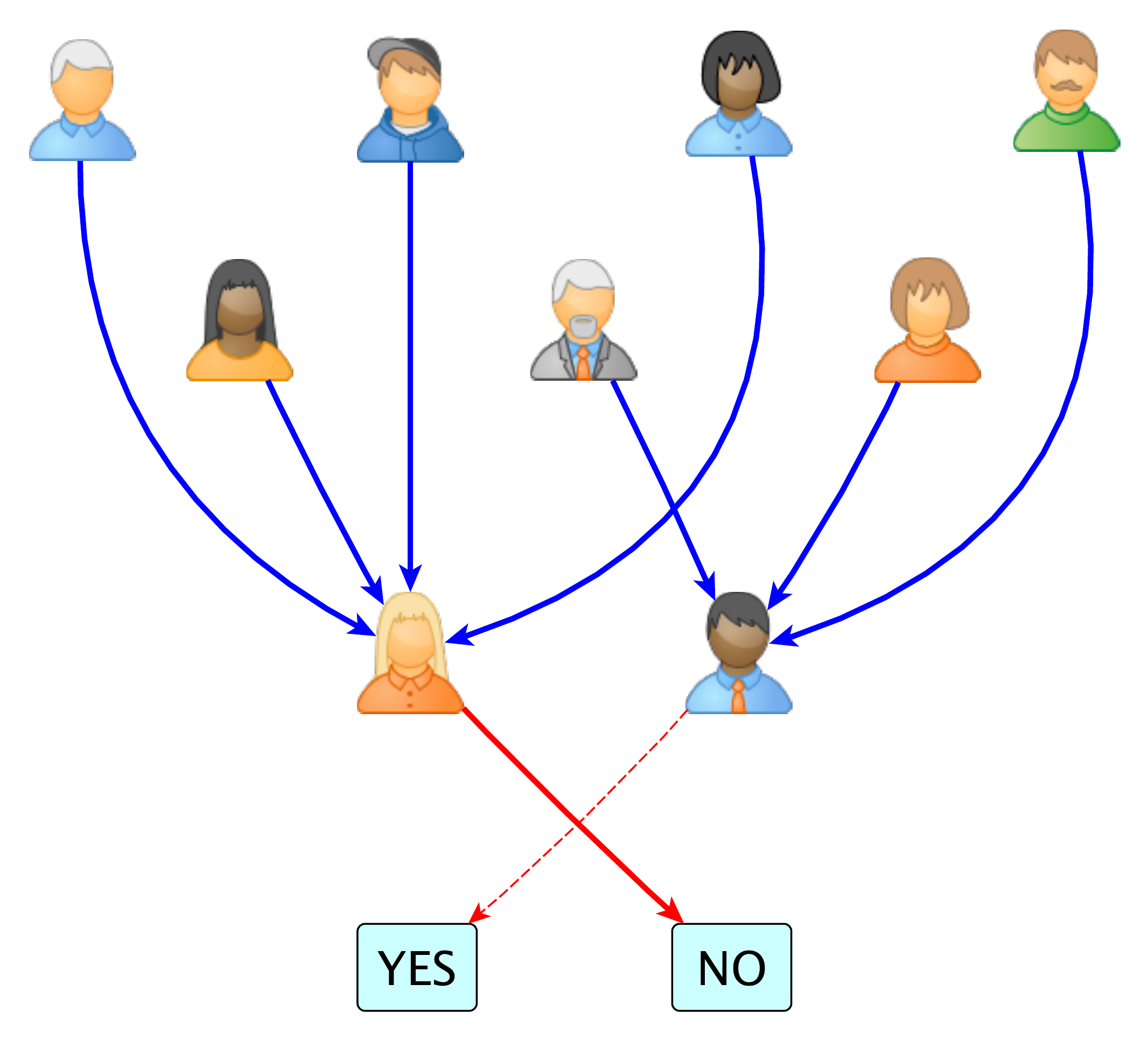} & &
  \includegraphics[width=0.2 \textwidth]{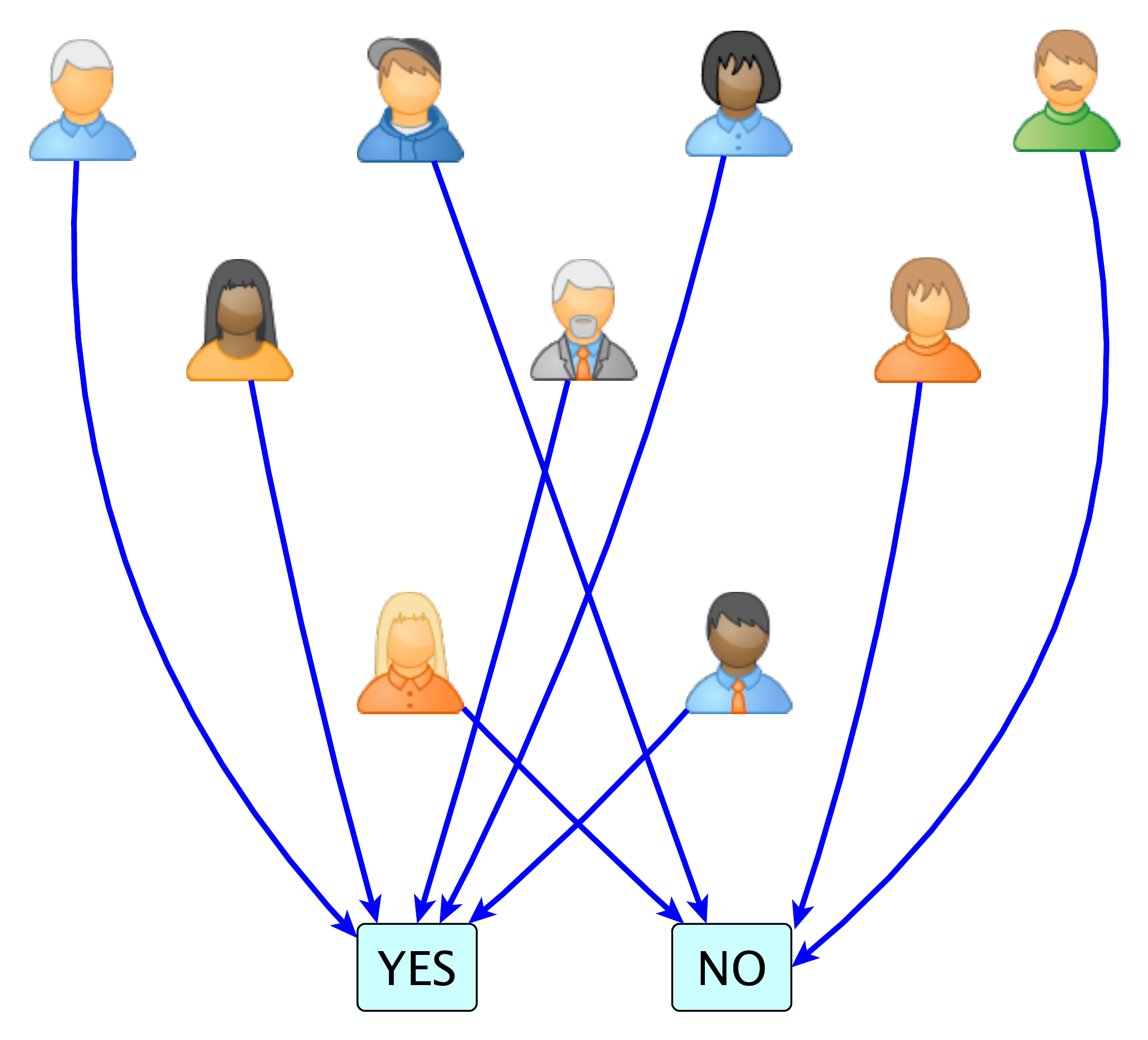} & &
  \includegraphics[width=0.2 \textwidth]{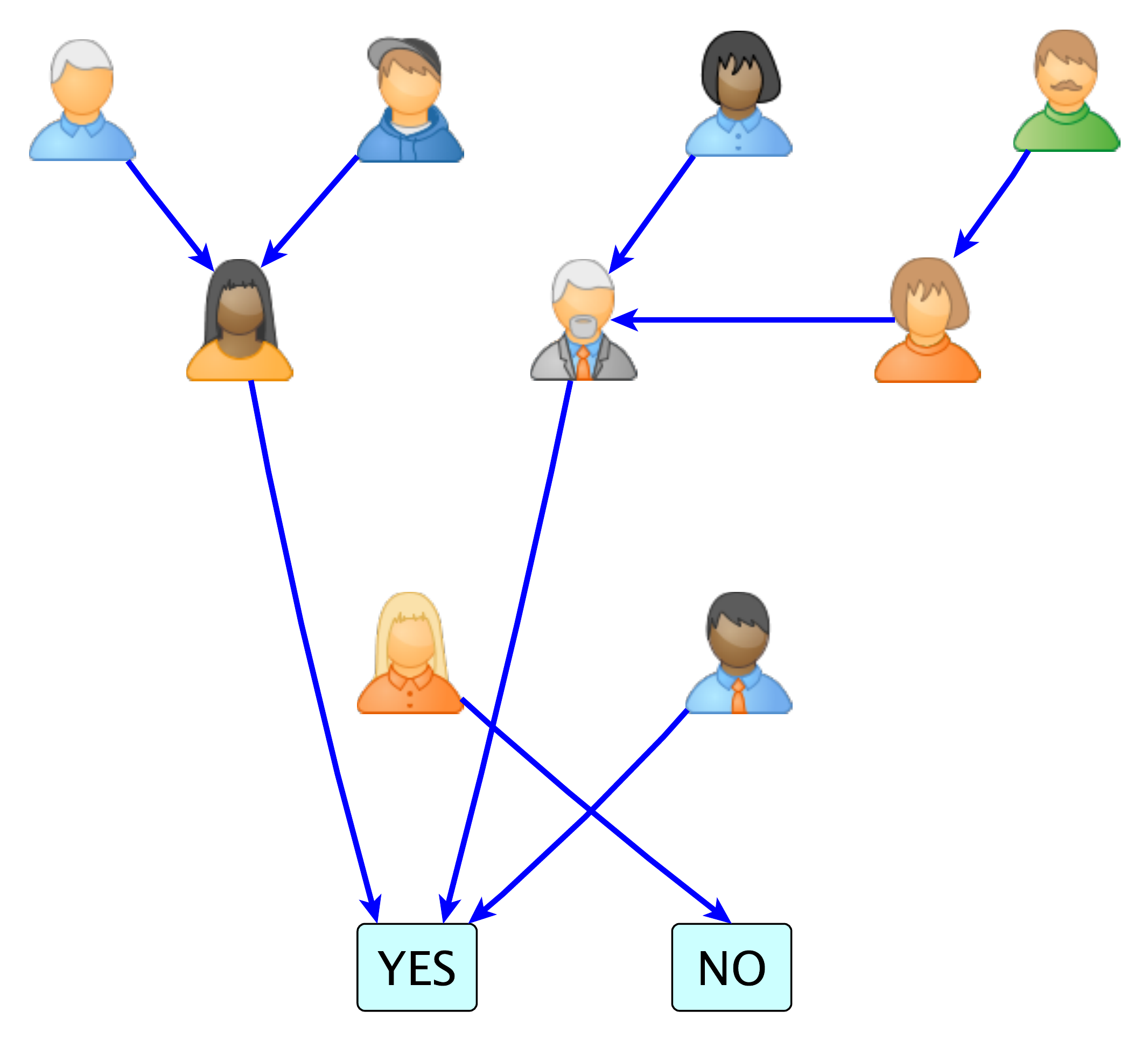}
 \end{tabular}
\\

\captionof{figure}{Design differences in three types of democratic models: representative, direct and liquid democracy. In representative democracy, citizens are represented by elected officials who then vote on  behalf of the first. In direct democracy, citizens vote directly. In liquid democracy, every citizen either votes directly or delegates to a representative of her choice.}\label{fig:arch}
\end{center}

 \cite{blumzuber2016}  present two main normative arguments in favour of liquid democracy. First, allowing for delegation to more informed people allows for better decisions when there is an objectively ``correct'' decision. Second, delegation allows greater equality in the following sense. Instead of creating two types of citizens with unequal power to influence policy (the representatives who vote on policy and ordinary voters who can only vote for representatives), liquid democracy allows voters to either vote directly on policy or, delegate their votes to other agents who can vote directly on policy.
 
 
In this paper we focus on the first issue: We study the information aggregation properties of delegation in elections with two candidates and two states of nature, where truth-seeking and partisan voters co-exist. Truth-seeking voters (or independents) have state-dependent preferences over candidates, while partisan voters have fixed state-independent preferences over candidates.
Truth-seeking voters have some information regarding the state of the world, but their information precision might be heterogeneous (i.e. some voters might be better informed than others about the true state of nature), and partisans are divided into two (potentially unequal) groups, depending on which candidate they support. We investigate whether liquid democracy improves upon direct and representative democracy in binary setups in the sense of welfare (the probability that the ex-post welfare maximizing alternative is chosen) and in the sense of information aggregation (the probability that the alternative matching the state of nature is chosen).

A priori, it is not obvious if and when delegation improves outcomes. When the preferences and information quality of each voter is private information, it is possible for truth-seeking voters to unwittingly delegate to partisans or to other truth-seeking voters who are worse informed. Such delegation can hurt welfare and information aggregation. Even when the preferences and quality of information of all voters are commonly known, the answer is not obvious. In such settings, delegation introduces a tradeoff (\citealp{DBLP:conf/aaai/BloembergenGL19,kahng2018,halpern2021defense,Armstrong21}): while delegation from a poorly informed voter $i$ to a better informed voter $j$ increases the electoral power of  $j$ (which is desirable), it leaves $i$ without a vote to express her own private information and hence leaves valuable information out of the aggregation process (which is undesirable). Whether the net effect of delegation is positive or negative, depends on the exact behavior employed by the voters (\citealp{kahng2018,halpern2021defense}). To our knowledge, the literature has not yet determined the net effects of this tradeoff; we investigate how strategic voters resolve it in equilibrium.


Our main contribution is to show that delegation can improve information aggregation in finite collective decision making  of any size, both when the preferences and the information precisions of voters are known, and when they are not. 

When voter types are common knowledge, we show that liquid democracy admits a better equilibrium---in terms of information aggregation and welfare---than direct and representative democracy. In general, the best such equilibrium is hard to characterize, but we provide some of its properties and characterize it in the case of an electorate containing sufficiently many uninformed voters. In order to tractably study how voters can use delegation to improve aggregate information in other cases, we then focus on a class of committees in which all voters are truth-seeking and are either nonexperts (who possess information of homogenous quality) or experts (who are better informed than nonexperts). 
Under some conditions, we characterize how voters solve the tradeoff posed by delegation optimally amongst a natural class of equilibria.\footnote{We look at \emph{neutral} equilibria in which voters adopt strategies that are symmetric with respect to the state. This equilibrium selection is natural in environments which are symmetric with respect to the state.} Efficient information aggregation requires some voters delegate to better informed voters such that the number of votes held by non-delegators optimally reflects the precision of their information. Hence, under the efficient equilibrium, it is  worth wasting some information (due to delegation) in order to form a sub-electorate of decision makers who are `optimally weighted' relative to their information qualities. 

While collective decisions can perform strictly better with delegation than without, as we show, 
liquid democracy also admits inefficient equilibria. These inefficient equilibria can do worse than the best equilibria of direct democracy and representative democracy, meaning the comparison between the mechanisms is not unambiguous due to the plethora of equilibria. Indeed, it is well known that multiplicity of (inefficient) equilibria is pervasive in voting games, where there are many situations in which no voter is pivotal. There can often be numerous Nash equilibria in which players use, for instance, weakly dominated strategies. We show that delegation can exacerbate this problem of multiplicity; for instance, additional inefficient equilibria in liquid democracy arise when there is `too much' delegation to a small set of voters. The best equilibria in liquid democracy requires coordination between voters on who delegates to whom; we show that mis-coordination can create incentives for voters to delegate much more than what is efficient. This concern of vote concentration in liquid democracy has been noted by \cite{kahng2018}, by \cite{Campbell2022} who find subjects inefficiently overdelegate in an experiment on liquid democracy, and in practice (the German Pirate Party found that a few `celebrities' were delegated large numbers of votes when implementing liquid democracy)\footnote{\label{new_scientist}https://www.newscientist.com/article/mg23531424-500-bitcoin-tech-to-put-political-power-in-the-hands-of-voters/25}, but to our knowledge, it has not been discussed as an equilibrium outcome.

We then identify a condition under which issues of mis-coordination in liquid democracy may not be a huge concern. Eliminating weakly dominated strategies is often used to rule out implausible Nash equilibria in voting games. Taking this further in this context, Iterated Elimination of Weakly Dominated Strategies (IEWDS) is a powerful tool for predicting outcomes (\citealp{dhillon2004})\footnote{The concept of IEWDS has the drawback that the order of elimination may matter in reaching the dominance solvable outcome. However, when voters have strict preferences over alternatives, as in our setting, \citet{marx2000} show that the order of elimination does not affect the outcome. Finally, a common knowledge justification for IEWDS was provided by \citet{rajan1998}.}. We show that the game with delegation is dominance solvable as long as at least one truth-seeking voter is sufficiently well informed, while the game without delegation (i.e. direct democracy) is not. Importantly, the outcome of IEWDS is the efficient one and requires only two rounds of iterated elimination, i.e. it is cognitively easy. The IEWDS solution also only relies on type-symmetric strategies; together this suggests that the presence of a sufficiently well informed voter may make it easy for voters to coordinate on the efficient equilibrium in liquid democracy. In contrast, reaching efficiency under direct democracy may be much harder for voters as the game is not dominance solvable even in the simplest settings and efficiency may require voters of the same type to employ asymmetric strategies. In absence of a coordination mechanism, electoral accidents may be more likely to occur under direct democracy. Finally, it is worth noting that voter participation is greater under the IEWDS solution with delegation than under the efficient equilibrium without delegation. With delegation, all votes are cast either directly or indirectly (i.e. via delegates). When delegation is not allowed, substantial abstention takes place in every efficient equilibrium even when voting is costless (\citealp{swing1996}).
 
What happens when the voters' types are their private information? In this case, vote delegation introduces additional trade-offs, leading to a substantially more involved analysis. On the one hand, liquid democracy can increase the vote-share of the efficient alternative if less informed voters transfer their ballots to better informed truth-seeking voters; but it may also harm the electoral prospects of the efficient alternative if uninformed voters delegate their votes inadvertently to partisans. Indeed, since the players types are unobservable, mistakes in vote transfers are unavoidable.
Again, we prove that the best undominated equilibrium in liquid democracy is always at least as good as the best undominated equilibrium in direct and representative democracy. This is a very strong result, since it holds for every possible type of distribution and society size, and establishes that, despite the additional trade-off, delegation is  welfare improving.\footnote{In settings composed only of truth-seeking voters, it is known that communication before voting can also boost information aggregation (e.g. \citealp{gerardi2007deliberative}). However, it is not clear that it would do so when there is also a group of partisan voters---or, merely, a chance that such voters exist---in the electorate.}

Overall, our main findings combined strengthen the case for liquid democracy. To our knowledge, this is the first paper to establish that rational truth-seeking voters can always exploit delegation to improve their welfare, when partisan voters also exist in the electorate. However, liquid democracy does not guarantee that better outcomes will prevail, as the worst equilibrium of liquid democracy is  worse than the best of direct or representative democracy. Despite that, liquid democracy can lead to efficient information aggregation in some cases without the demanding cognitive requirements (e.g. with respect to higher order beliefs) necessary for direct democracy to aggregate information well (see, for instance, \citealp{swing1996}). This indicates that liquid democracy deserves further academic exploration (e.g. by means of laboratory and field experiments---see, also, \citealp{Campbell2022}). 

The rest of the paper is organised as follows: Section \ref{sec: lit} discusses the related literature, Section \ref{sec: prelim} presents the definitions and notation we use, Section \ref{sec: pub} deals with the case when information on voter types is common knowledge, Section \ref{sec: dominance} discusses dominance solvability when types are common knowledge, Section \ref{sec: private} analyses the case when voter types are  private information, and we conclude and discuss potential future avenues of research in Section \ref{sec: conclusion}.

\section{Related Literature} 
\label{sec: lit} 
There is a vast literature on information aggregation in two candidate elections starting with the seminal work on the Condorcet Jury theorem (1785), which showed that, with two alternatives, two states of the world and common values, if each individual voted for the correct alternative with probability---same across subjects---strictly greater than half, then  the probability that a majority would choose the correct alternative  $(1)$ is higher than the probability that any one individual would, and $(2)$ converges to one as the society grows large. The theorem assumed sincere voting and follows from the law of large numbers. \cite{austen1996} showed that sincere voting was not always rational in such a setting.  \cite{mclennan98} (for common value elections) and \cite{FeddersenPesen1997} (with heterogeneous in information voters and private information on voter types)  show that,  for two candidate elections, there is always an equilibrium that aggregates information efficiently, asymptotically as the size of the electorate goes to infinity. They do not allow abstention. 
 \cite{swing1996},   allow for heterogeneous (in preference) voters who are either fully informed or fully uninformed,  private information and allow abstention. In their setting where the size of the population is unknown, yet large, and truth-seeking and partisan voters co-exist; they show that elections aggregate information efficiently, i.e. the equilibrium outcome is asymptotically the same as though information on the state were common knowledge. In equilibrium, informed truth-seeking voters and partisan voters vote for the alternative that they support, while uninformed truth-seeking voters employ  non-trivial behavior.  More specifically, they are subject to a ``swing voters curse'': a substantial fraction of voters abstain strategically to allow informed truth-seeking voters to be decisive, even when voting is costless. For their main result, they do need large elections.\footnote{As shown by \cite{shotts2006signaling} and \cite{meirowitz2009pivots} moderate voters need not only abstain to avoid diluting the informativeness of the election, but also to signal their private preferences to the politicians in settings of repeated elections.} 
In contrast, we show that delegation can improve on simple majority voting with abstention in elections of any size, with and without private information on voter types.

Also close to our paper is work in the weighted majority voting (WMV) literature (\citealp{nitzan1982, shapley1984optimizing, ben2015voting}). Papers in this literature often study voters making a binary decision through majority voting and wanting this decision to match some underlying state of the world. Typically, voters have heterogeneous qualities of information about this underlying state, and this literature asks (among other questions) what the optimal weight (or number of votes) a social planner should \emph{exogenously} assign to each voter in order to maximize the probability the election outcome is correct. The key difference between  this literature and liquid democracy, is that in liquid democracy voters obtain different numbers of votes \emph{endogenously}.\footnote{Bouton et al. (2022) analyze a model each voter is assigned a number of votes and decides instrumentally how many of them she wants to use (i.e. she is allowed to partially abstain). As we show in section 6, allowing agents to delegate voting power has distinct effects from allowing partial abstention, and can lead in certain circumstances to higher welfare.} The results in certain common-interest special cases of liquid democracy relate closely to those in the WMV literature; we discuss the connection with \citet{nitzan1982}, in particular, in section 4.

 \cite{DBLP:journals/corr/ChristoffG17} and \cite{kahng2018} also study information aggregation with delegated voting, but do not employ an equilibrium approach. 
 \cite{DBLP:journals/corr/ChristoffG17}
focus on the aggregation of individual choices to  social choice and the unintended effects of delegation on the rationality postulates satisfied by direct voting. \cite{kahng2018} study the case where voters have different levels of information, complete information and a network setting. They show that when voters delegate only to more informed voters who are within their local network, then delegation can lead to worse outcomes than simple majority voting, due to the concentration of power. They argue that if alternative behavioral assumptions are imposed then delegation can lead to better outcomes. The key point of the paper is that too much delegation to the same voters (given the network structure) risks losing out on valuable information. \cite{halpern2021defense}, meanwhile, study a setting similar to that of \cite{kahng2018} and find sufficient conditions satisfied by examples of delegation mechanisms under which LD must improve upon DD in large committees. In both these papers, voters are not strategic and behave as perscribed by a given mechanism. We complement their work by showing that when  behavior rules are not fixed but rather endogenously determined by rational voters then delegation can lead to welfare improvements in equilibrium but need not do so.   \cite{DBLP:conf/aaai/BloembergenGL19}, study equilibrium behaviour in a network setting. Voters  know which of two outcomes they would prefer with a probability between 0.5 and 1 (voter accuracy).   They focus on the decision problem of voters, when voters have a choice of direct voting, which incurs a cost (of getting information), and delegated voting, when they know the accuracy type of the other voters and the probability that they have a similar preference, but they do not know the true type. The authors show existence of Nash equilibria and average accuracy achieved in a network setting. In contrast, we study simultaneous delegation and costless voting and we include abstention, complete and incomplete information settings.
 While we focus on information aggregation, \citet{green2015direct} also studies a non-strategic model of liquid democracy but focuses on potential preference aggregation benefits of the system: unlike in many implementations of representative democracy, voters can delegate to someone who has close preferences to them, issue by issue. 
Recently, \cite{Armstrong21} in a pure common interest and complete information game, with costly voting in a fully connected network, show theoretically that delegation always reaches an equilibrium with weakly higher group accuracy at identifying the "ground truth" (objectively correct) outcome.  In their model, there is a first stage of delegations in which every voter can observe the delegations that have already taken place, and which ends only when no voter can improve the probability of selecting the correct alternative by further delegation. This is what ensures that only those delegations, which improve the chance of a correct outcome, happen.  In contrast, our model has partisan voters, has costless voting, simultaneous delegation and voting, and we also handle the incomplete information setting which might be more realistic.\footnote{Notice that, introducing a cost to voting makes it more attractive to delegate.}
However, their experiments with machine learning classifiers show that
neither optimal delegations nor efficiently computable delegation strategies significantly improve accuracy in small or realistically sized electorates, respectively. 
 \cite{Campbell2022} study a model of liquid democracy both theoretically and experimentally. Their theoretical model is nested by ours (modulo some technical details). Theoretically, they characterize an important class of liquid democracy equilibria in specific types of common value electorates, with no abstention; we discuss their theoretical results in Section 6. Experimentally, they find that despite the theoretical possibility of liquid democracy to improve outcomes, subjects systematically delegate too much and do worse with liquid democracy than simple majority voting. While we also discuss potential issues of overdelegation in liquid democracy (Section \ref{sec: pub}), in Section \ref{sec: dominance} we describe a setting where it is cognitively easier for voters to coordinate when delegation is possible; such settings may improve the odds for liquid democracy to perform well in  experimental testing.


Beyond the specific issue of delegation, our work also relates to studies which try to assess the information aggregation properties of different electoral systems, most of which are conducted in an asymptotic setup (i.e. they are relevant only for large elections, while most of our results pertain to groups of any size). \cite{bhattacharya2013preference} and \cite{Barelli2022} extend the Condorcet's Jury theorem to heterogeneous state-dependent preferences and to general state and signal spaces respectively, and show that information aggregation depends on the complexity of the preference and information structure. \cite{Goertz2011}, \cite{Bouton2016} and \cite{Ahn2016} study the properties of approval voting and other scoring rules, for any electorate size, when a divided majority occurs due to disagreements among the truth-seeking voters regarding the most likely state of the world. These papers focus on the case where the majority needs to coordinate in order to eliminate the possibility of inefficient outcomes, and show that approval voting performs better than plurality in settings where truth-seeking voters are identical in terms of information precision. As far as runoff systems with two (or more) voting rounds are concerned, \cite{Piketty2000} observes that due to the existence of multiple voting rounds, the majority group should have additional opportunities to coordinate and aggregate their information. \cite{Martinelli2002} shows that efficient aggregation of information is feasible in equilibrium under a two-round runoff rule, in the setting of a divided majority with three alternatives. \cite{Tsakas2021} extend this possibility result to more general settings. \cite{herrera2019marginal} assess theoretically the information aggregation properties of more ``proportional'' systems. That is, when  a small change in the vote-share distribution affects the outcome, even if it does not affect the winner of the election. They find that in large societies, relatively uninformed voters abstain, and information is aggregated efficiently. Finally, \cite{mcmurray2017voting} and \cite{prato2017wisdom} show that voting might be less efficient in aggregating information, when the policy alternatives are proposed by self-interested candidates.  Our work contributes to this literature by showing that, beyond the ballot space, the aggregation method, and the process through which alternatives emerge, the possibility of endogenous vote transfers can crucially affect the quality of electoral outcomes, even in electorates with finite size.

\section{A model of Liquid Democracy} \label{sec: prelim}

Consider a set of voters $\mathcal{N} = \{1,...,N\}$ who vote to make a
single binary decision between two alternatives $C=\{A,B\}$. The alternative
receiving a simple majority of votes cast is the election outcome $%
\mathcal{O} \in \{A,B\}$ (with ties broken uniformly at random). There is an unknown state $\omega \in \{a,b\}$ and
voters have common prior probability $Pr(\omega|\omega=a)=\pi$, with $\pi \in (0,1)$. 

\textbf{Types of voters.} At the start of the game, voter $i$'s type $t_i$
is realized (we consider types being privately or publicly observed).
After this, $i$ receives a private signal $s_i \in \{a,b\}$ which has a
distribution depending on her type. Voters' types consist of two components: (1)
their preferences over the election outcome under full information, and (2)
the quality of their private information about $\omega$.

A voter $i$ has preferences $p_i \in \{A,B,I\}$, where $A$ indicates preference for alternative $A$, $B$ for alternative $B$, and $I$ for being `independent'. If 
$p_i=I$, then voter $i$ prefers that  the election outcome matches the state of the world and gets an ex-post
payoff equal to 1. The ex-post payoff of voter $i$ is denoted by 
$$u_{p_i}(\mathcal{O},\omega),$$
where $\mathcal{O}$ is the election outcome, $\omega$ the state of the world, and $p_i$ the type of voter $i$. If 
$p_i=I$, then the independent voter $i$ prefers that  the election outcome matches the state of the world and gets an ex-post
payoff equal to 1, that is
$u_{I}(A,a)=1$ and $u_{I}(B,b)=1$.
Voters
with $p_i \in \{A,B\}$ are `partisans' and get an ex-post payoff equal to 1 when the election outcome matches their preference regardless of the state, i.e.  $u_{A}(A,\omega)=1$ and  $u_{B}(B,\omega)=1$.


A voter $i$ has information quality $q_i \in [0.5,1]$ drawn from a commonly
known distribution $Q_i$ (note that $q_i$ is independent of $p_i$).\footnote{%
It doesn't matter that $q_i$ is independent of $p_i$. If not, all we would
care about is the distribution of $q_i$ conditional on $p_i=I$.} A higher $%
q_i$ means the voter is more likely to receive the correct signal since $%
Pr(s_i=\omega|\omega)=q_i$, where $q_i$ is the precision of voter $i$'s
information.\footnote{%
We  assume that signal precision is the same in both states.} If $q_i=1$,
the voter is perfectly informed and, if $q_i=0.5$, the voter is then perfectly
uninformed. Voter $i$'s information precision $q_i$ is  irrelevant when it comes to partisan voters
who prefer one alternative regardless of the state.

A voter $i$'s type is given by $t_i=(p_i,q_i)$, where $ (p_i,q_i) \in T = \{A,B,I\} \times [0.5,1]$. Then we denote a
profile of all voters' types as $t=(t_1,...,t_N) \in T^N$.

\textbf{Private signals.} After types are realized (and observed
privately or publicly), each voter $i$ receives a private signal $s_i$
with information precision $q_i$.

\textbf{Strategies.} We call a voter's choice of what actions to take after
her type is realized an \emph{interim strategy}. This type of strategy is useful when types are publicly observed. A full contingent plan of what
actions a voter takes as a function of her realized type, which is
useful when types are private, is an \emph{ex-ante} strategy. We next describe
these types of strategies.

\textbf{Interim Strategies.} When delegation is not allowed, a (pure)
interim strategy for a voter $i$ is a choice of action as a function of $i$%
's signal realization chosen after her type $t_i$ is realized. When delegation is not
allowed, an interim strategy for $i$ is given by the mapping $\sigma_i: \{a,b\}
\rightarrow \{a,b,x\}$. Function $\sigma_i$ maps  $i$'s signal realization $s_i$
to a decision on whether to cast her vote either in favor of alternative $A$, alternative $B$, or
to abstain (the choice denoted by $x$). If $\sigma_i(a)=a$ and $\sigma_i(b)=b$, we say voter $i$
votes \emph{sincerely}, while if $\sigma_i(a)=\sigma_i(b)=a$, we say $i$ votes \emph{unresponsively} for $a$ (voting unresponsively for $b$ is defined similarly). When delegation is allowed, an interim strategy for $%
i$ is a mapping $\sigma_i: \{a,b\} \rightarrow \{a,b,x,d_{j \neq i}\}$,
where the choice of $d_{j \neq i}$ is a voter $j$ that $i$ chooses to delegate
to. If $i$ delegates to $j$ and $j$ does not delegate, then $i$'s vote is cast for whichever
alternative $j$ votes for (or neither if $j$ abstains). If $j$
delegates her vote to another voter $k$, $i$'s vote is delegated to $k$ along with $j$'s (that is, we allow for transitive delegation). We call  a \emph{delegation cycle} any sequence of voters $i_1,...,i_K$, with $K \geq 2$, such that each voter $i_k$, with $1 \leq k < K $, delegates to voter $i_{k+1}$ and $i_K$ delegates to $i_1$.  In the case of a delegation cycle, all votes held by the voters in such a cycle are cast in abstention. \footnote{None of our results depend on allowing for transitive delegation or nullifying all votes cast in a delegation cycle. While changing either assumption may change the set of equilibria of the game, they will not affect our main results concerning LD having a better equilibrium than DD or RD, nor will they affect our results characterizing best equilibria and best neutral equilibria in Section 4.} We assume that if a voter votes for an alternative (or abstains), all votes she holds are cast for this alternative (or in abstention).\footnote{Similar to the previous footnote, none of our results would change if voters were allowed to split the votes they hold amongst multiple alternatives. However, allowing for this would require strategies to allow for complicated conditional behavior as delegation and voting happen simultaneously.}

\textbf{Ex-ante Strategies.} To describe a voter's strategy before their
type is realized, we need them to choose a contingent plan for every
possible type realization. An ex-ante strategy for voter $i$ 
is a mapping $\Sigma_i: T \times \{a,b\} \rightarrow \{a,b,x\}$ when
delegation is not allowed and is a mapping $\Sigma_i: T \times \{a,b\}
\rightarrow \{a,b,x,d_{j \neq i}\}$ when delegation is allowed. Voter $i$ chooses an interim strategy
for every possible type they could be. For any ex-ante strategy of voter $i$, 
$\Sigma_i,$ and any possible realized type of $i$, $t_i$, let $\Sigma_i(t_i)$ be the
interim strategy of voter $i$ with type $t_i$ under strategy $\Sigma_i$.

\textbf{Timing.} The timing of the game is as follows: \begin{enumerate}
    \item Nature draws state $\omega$.
    \item Types: if there is complete information on types, the vector of types $(t_1,...,t_N)$ is commonly observed. If there is incomplete information on types, voter $i$'s type $t_i=(p_i,q_i)$ is drawn from $P_i \times Q_i$ and is privately observed.
    \item Voters' private signals are realized.
    \item Voters choose to vote/abstain/delegate. If there is complete information on types, they do so following interim strategies; with incomplete information on types they follow ex-ante strategies.
    \item The election outcome is determined.
    \item Payoffs are realized.
\end{enumerate}

\textbf{Strategy Profiles.} An interim strategy profile $\sigma=(%
\sigma_1,...,\sigma_N)$ is a vector of interim strategy profiles for each
voter. Similarly $\Sigma=(\Sigma_1,...,\Sigma_N)$ is an ex-ante strategy
profile. For any ex-ante strategy profile $\Sigma$ and profile of realized
types $t$, let $\Sigma(t)=\left(\Sigma_1(t_1),...,\Sigma_N(t_N)\right)$ be the profile
of interim strategies implied by $\Sigma$ at $t$.

\textbf{Expected utilities and equilibrium.} Given a publicly observed
type profile $t \in T^N$ and an interim strategy profile $\sigma$, let $%
U_i(\sigma,t)$ be voter $i$'s expected utility from voters following interim
strategy profile $\sigma$ and $U_i(\sigma_i^{\prime },\sigma_{-i},t)$ be the
same expected utility when voter $i$ deviates to interim strategy $\sigma_i^{\prime }$.

When types are privately observed, we are interested in  ex-ante
expected payoffs. Fixing ex-ante strategy profile $\Sigma$, let 
\begin{align}
EU_i(\Sigma) = \sum_{t \in T^N} \left( U_i\left(\Sigma(t),t\right)\cdot Pr(t) \right)
\end{align}
be $i$'s expected
utility and $EU_i(\Sigma_i^{\prime },\Sigma_{-i})$ be the same utility when $i$
deviates to $\Sigma_i^{\prime }$.

We look for Bayes Nash equilibria (BNE) of the game in interim strategies
when types are publicly observed and in ex-ante strategies when types are
private information. In both the publicly and privately observable types' cases, we will select for BNE in which weakly dominated strategies are not employed; when we say `equilibrium', henceforth this means BNE in undominated strategies. As it is typical in voting games, there  typically are many equilibria; ruling out equilibria in dominated strategies  alleviates this issue.

\textbf{Direct and Representative Democracy. }Our main aim is to
compare equilibrium outcomes of liquid democracy, with equilibrium outcomes of direct (DD)
and representative democracy (RD). To this end, we need to provide a definition
of these two alternative systems, which formally differ from liquid 
democracy (LD) only in terms of the involved interim strategies. In DD, an interim strategy for voter $i$ is given by mapping $\sigma
_{i}:\{a,b\}\rightarrow \{a,b, x\}$. 

RD differs from DD and LD in that there is a fixed set of representatives $\mathcal{J}$. Set $\mathcal{J}$ can contain any subset of voters in $\mathcal{N}$ and additionally contains two representatives $a^*$ and $b^*$, where $a^*$ is a partisan representative, who is an advocate for alternative $A$ and we assume always (mechanically) votes for $a$, and $b^*$ is an advocate for $B$ who always votes for alternative $B$. Representatives $a^*$ and $b^*$ are hence not strategic and are  not included in $\mathcal{N}$. In RD, an
interim strategy for $i\notin J$ is a mapping $\sigma
_{i}:\{a,b\}\rightarrow \{x,d_{j\in J}\}$ and for $i\in J$ is a mapping $\sigma _{i}:\{a,b\}\rightarrow \{a,b, x\}$.\footnote{Other than the existence of $a^*$ and $b^*$, Representative Democracy (RD)  differs from Liquid Democracy (LD)  in two ways:  first the potential delegates are fixed exogenously and second, if a voter delegates to a representative this delegation applies for all issues/decisions faced by the committee. Liquid Democracy meanwhile allows delegates to be different for different issues. However in this paper we only focus on the first difference, as our model concerns a single issue.} The assumption that representatives $a^*$ and $b^*$ always exist, vote for their preferred alternative, and are not included in $\mathcal{N}$ is for technical convenience in comparing RD with DD and LD in common interest committees.  

In summary, under direct democracy DD, agents can assign votes either to $a$ or to $b$
(or abstain), while under representative democracy they can assign votes to a fixed set of representatives. 


\textbf{Comparing mechanisms.} There are multiple measures of interest in comparing liquid democracy with representative and direct democracy. One natural measure of the welfare attained under an equilibrium of a mechanism is the probability it implements the `majoritarian outcome' (i.e. the alternative that is preferred ex post by most voters). Another important measure is the probability the mechanism implements the `correct' outcome: i.e. the outcome that matches the state of the world. Implementing the correct outcome is not only important to independent voters, but it also gives a measure of how well the committee aggregates voters' private information.

Note that if an equilibrium of a mechanism implements the correct alternative with higher probability than another equilibrium of the same or of a different mechanism, then it also implements the majoritarian outcome  with higher probability. This is so because for all type draws such that the truth-seeking voters are not pivotal as a group (i.e. their choice cannot affect the outcome when players use undominated strategies), then both equilibria lead to the majoritarian outcome with certainty and to the correct outcome with the same probability. Observe that in the remaining type draws the correct alternative coincides with the majoritarian outcome. Therefore, since one equilibrium implements the correct alternative with higher probability than the other unconditionally, it must be the case that it implements the correct alternative---and hence also the majoritarian outcome---with higher probability conditional on the truth-seeking types being decisive. Since the majoritarian outcome is implemented with higher probability under the former equilibrium than the latter conditional on the group of truth-seeking voters being decisive, and with equal probability conditional on truth-seeking voters not being decisive, it leads to the implementation of the majoritarian alternative with higher probability unconditionally. We call the equilibria that maximize the probability of the majoritarian outcome as {\it best} equilibria. 

In the next sections, we compare LD with DD and RD in terms of the best and worst equilibria under the three voting sytems. We first analyse the complete information setting and then the incomplete information setting.    

\section{Complete information about types} \label{sec: pub}

In this section, we assume each agent's type is common-knowledge. We are interested in the structure and information aggregation properties of equilibria; of particular interest are equilibria which best aggregate information (and hence implement the majoritarian outcome with the highest probability). We compare the best equilibria in undominated strategies  in  LD with both DD and RD. While we cannot characterize the set of equilibria, we show the best equilibrium of LD does better than that of DD and RD, and provide some examples of committees in which LD does strictly better. Below we provide the intuition for our first result (Proposition \ref{prop_bestequil}).


 Let $n_I$ be the number of independent voters and $n_A$ and $n_B$ be the number of partisans for $A$ and $B$ respectively. Assume $n_I \geq 1$ or the problem is uninteresting. The existence of an equilibrium in undominated strategies is straightforward in LD, DD and RD when types are commonly known. First, in LD and DD it is a weakly dominant strategy for every partisan $i$ with $p_i=A$ ($p_i=B$) to vote for $A$ ($B$) regardless of her signal. This is simply because the partisan strictly prefers this strategy when pivotal and is indifferent across all strategies otherwise; moreover, every voter is pivotal with probability $>0$, given some strategy profile of her opponents.\footnote{
E.g. w.l.o.g consider an A partisan. If $N$ is even, let $\frac{N-2}{2}$ voters vote for $A$ and $\frac{N}{2}$ vote for $B$. If N is odd, then let $\frac{N}{2}$ vote for $A$ and  $\frac{N}{2}$ vote for $B$.} Given this, in every equilibrium of LD and DD, each partisan votes unresponsively in favor of her preferred alternative.  Meanwhile in RD, $a^*,b^* \in \mathcal{J}$ always vote (mechanically) for their preferred alternatives. Other partisans in $\mathcal{J}$ have weakly dominant strategies of voting for their preferred alternative, and partisans not in $\mathcal{J}$ have a weakly dominant strategy of delegating to either $a^*$ or $b^*$. Partisans will adopt these strategies in any equilibrium in undominated strategies. Having pinned down the behavior of partisans in all three mechanisms, in the case that $|n_A-n_B|>n_I$, in any equilibrium independent voters are (collectively) never pivotal and hence any combination of non-weakly dominated strategies they use constitutes an equilibrium. The election outcome in equilibrium in this case is deterministic. The interesting case is hence when $n_I \geq |n_A-n_B|$. Having fixed the equilibrium behavior of partisans, from the perspective of independents, the game is common-interest. 
As types are commonly known and strategy spaces are finite, there is a profile of non-weakly dominated strategies that maximizes payoffs for independents and, by McLennan (1996), this profile constitutes an equilibrium for each of the LD, RD and DD systems. 
Note that every strategy that is feasible in DD/RD is also feasible in LD and leads to weakly lower welfare compared to the strategy profile that maximizes the probability of implementing the correct alternative in LD---which constitutes an equilibrium.

\subsection{Liquid Democracy}
\label{sec: LD}

The best equilibrium for independents  in LD is of interest beyond just proving existence; it is also the equilibrium that best aggregates the information of the committee. In general, we cannot pin down this strategy profile; the following Proposition tells us that such an equilibrium in LD does better than the corresponding equilibrium in RD/DD.

\begin{prop} \label{prop_bestequil}
In liquid democracy with commonly known types, there is a best (in terms of information aggregation) equilibrium $\sigma$ which does at least as well as the best equilibrium in direct and representative democracy.
\end{prop}

The proof follows from the logic at the end of the previous subsection. We note that whenever the best equilibrium of liquid democracy requires use of delegation, this best equilibrium will do strictly better than that of direct democracy. The comparison between liquid and representative democracy is a bit more nuanced. It can be possible to implement the best equilibrium of liquid democracy under representative democracy if the set of representatives $\mathcal{J}$ contains exactly the set of voters who are delegated to under the best equilibrium of liquid democracy. If some voters who are delegated votes under liquid democracy's best equilibrium are not contained in $\mathcal{J}$, then representative democracy will do strictly worse than liquid democracy. Meanwhile, if the set of representatives $\mathcal{J}$ is too large then representative democracy can also fail to do as well as liquid democracy. For instance, if $\mathcal{J}$ contains all voters in the committee, then representative democracy is equivalent to direct democracy. Hence the welfare comparisons of liquid and representative democracy will depend on the specific structure of the committee.

For some simple types of committees, we can pin down the best equilibrium in LD. This will give some insight into the role delegation plays in information aggregation. First, we introduce a result from the weighted majority voting (WMV) literature, which helps us to characterize equilibria for special cases.

\textbf{Connection to Weighted Majority Voting (WMV):} 
The WMV literature assumes all voters are independents and considers binary state models in which voters make a binary decision seeking to match the state. Instead of allowing for delegation before voting, voters are assigned weights (a number of votes which can be any weakly positive real number) exogenously. Voters receive private signals with varying precisions and cast all of their weight (or votes) sincerely. The alternative $A$ or $B$ receiving the majority of votes  wins and voters prefer this outcome to match the state of the world. 

An important result (found in multiple papers starting with \citet{nitzan1982}) is that there exist `optimal' weights for each voter $(w^*(1),...,w^*(N)) \in \mathbb{R}_{++}^{N}$ such that the following holds: Under these optimal weights and under sincere voting, the outcome of the election is efficient -- or first-best -- in the sense that it is always the same outcome as if a single decision maker (with the same preferences as the voters) looked at all private signals and made the decision unilaterally. When voters are assigned these optimal weights, the committee \emph{efficiently aggregates} the private information of all committee members in the sense that the election outcome is the first-best (absent asymmetric information) outcome with probability equal to $1$. \citet{nitzan1982} show that for $i \in \mathcal{N}$, the optimal weight is $w^*(i)  = \log(\frac{q_i}{1-q_i})$, where $q_i$ is the precision of voter $i$.\footnote{Note that if $q_i=1$ then the optimal weight is infinite. This is because voter $i$ always knows the state and hence should optimally dictate the election outcome.} Under these weights, each voter's power optimally matches the quality of their information. If, for instance, $w^*(i) = 2w^*(j)$ for some $i,j$, this means that in terms of information content, voter $i$'s private signal is worth two of voter $j$'s private signal. If $s_i=a$ (voter $i$ receives information in favor of state $A$), then it would take more than two independent signal realizations of precision $q_j$ in favor of state $B$ to outweigh the information provided by $s_i$. Importantly for some of our results, these weights are not unique; scaling the weights of all voters by the same constant does not change voting outcomes.

In WMV, voters' weights are exogenously assigned; this means some voters are ex-ante more powerful than others, which may be arguably socially underisable or hard to implement. In liquid democracy all voters are ex-ante equally powerful (each is endowed with a single vote) but can use delegation to endogenously determine weights. However, under liquid democracy the optimal weights of \citet{nitzan1982} may not be attainable, as delegation results in the delegator being left without a vote. Despite this, some of our main results will show that it can be efficient for some voters to delegate to enable other voters to endogenously attain these optimal weights.

While, the best equilibrium under LD is typically hard to characterize, we can do so for an interesting class of committees.

\textbf{Committees with uninformed voters and few experts.} Consider a committee with prior $\pi = \frac{1}{2}$ and $n_A$ $A$ partisans and $n_B$ $B$ partisans (wlog assume $n_A \geq n_B$). Amongst independent voters, there are $n_U$ voters who are perfectly uninformed, that is $q_i=0.5$. There are $E$ independents who are `experts'; each such voter $i$ has precision $q_i>0.5$. The `experts' may be heterogeneous in their precisions. 

Suppose $n_U$ is large: there is a large number of uninformed independents. While these voters hold no information, they can aid other independents in aggregating information. First, as $n_A-n_B$ extra votes are  cast by partisans for $A$ over $B$, if $n_A-n_B$ uninformed voters vote unresponsively for $B$, they will entirely negate the partisans. There is no information loss in doing this. 

Given this behavior, weighted majority voting with the remaining uninformed voters, each having a weight of $0$ and each expert $i$ having a weight of $w^*(i)$, would deliver first-best outcomes. Attaining these outcomes only depends on the ratio of votes held by any two experts reflecting their relative information qualities: for any experts $i$ and $j$, the ratio of votes held should be $\frac{w^*(i)}{w^*(j)}$. One can see that delegation from uninformed voters (along with abstention from remaining uninformed voters) can arbitrarily approximate optimal weights as $n_U$ gets large. As the \citet{nitzan1982} weights are robust to certain small perturbations, we can then attain first-best information aggregation with LD for large enough (finite) $n_U$. Proposition \ref{prop_uninformed} captures this.

\begin{prop}\label{prop_uninformed}
Assume $\pi = \frac{1}{2},$ $n_A\geq n_B$, $n_U$ uninformed ($q_i=0.5$) independent voters and  $E$ independent experts ($q_i>0.5$). There exists $n^*$ such that when $n_U>n^*$, under the best LD equilibrium $\sigma$: \begin{enumerate}
\item $n_A-n_B$ uninformed independents vote unresponsively for $B$.
\item Each remaining uninformed independents either: (1) abstains at both signal realizations, or (2) delegates to the same expert at both signal realizations.
\item The election outcome under LD is the first-best outcome w.p. $1$.
\item Let $v_i$ and $v_j$ be the number of votes held by experts $i$ and $j$, respectively. Then, as $n_U \rightarrow \infty$, the ratio $\frac{v_i}{v_j}$ converges to the ratio $\frac{w^*(i)}{w^*(j)}$ (experts approach optimal relative weights).
\end{enumerate}

\end{prop}

Importantly, this result does not qualitatively depend on the worst informed voters being perfectly uninformed. In particular, whenever the result holds and the committee attains first-best outcomes in equilibrium, there exists an $\epsilon>0$ such that if all $n_U$ uninformed independents instead had precisions in the range $[0.5,0.5+\epsilon]$, then the same equilibrium would still deliver first-best outcomes. This is an implication of \citet{nitzan1982}'s  Corollary 1, which tells us that first-best solutions will anyway ignore the information of poorly informed voters who do not hold enough information to affect the best choice ex-interim. 

While the logic of the Proposition is that with large enough number of dispensable votes we can approximate \citet{nitzan1982} weights, the result is not an asymptotic one, as we can attain first-best outcomes for finite $n_U$. The following example demonstrates that in fact attaining first-best outcomes do not require $n_U$ to be very large.


\emph{Example.} There are four experts with $q_1=0.8, q_2=0.7, q_3=0.65, q_4=0.6$. The log likelihood ratios are respectively $1.3862, 0.8473, 0.618, 0.4055.$ Rounding off to 2 decimals we have $1.39, 0.85, 0.62, 0.4.$  Dividing through by 0.4 we have the normalised weights 3.475, 2.125, 1.55 and 1, or 3.5, 2.1, 1.5 and 1. 
Therefore 1 and $i$ should prevail over $j$ and $k,$ for any $i,j,k$ distinct from each other and from 1. However 2,3,4 should prevail over 1. 
This remains true even with rounded down weights. So we need $n_U\geq 3$ to achieve this result. With DD we cannot achieve FB-at best 4 and 5 can abstain.   

 Proposition 2 highlights two important roles uninformed independent voters can play in information aggregation: delegating to informed voters to enable them to achieve optimal weights, and voting unresponsively to counter partisans in the committee. As uninformed voters do not possess any private information, these two types of strategies are costless, i.e. they waste no information. This allows committees with enough uninformed independent voters to attain first-best outcomes. In the next section, we study committees in which all voters are informed (i.e. receive signals with $q_i>0.5$), which introduces a tradeoff in delegation: a voter $i$ delegating to a better informed voter can enable the better informed voter to have a weight closer to the optimal weight, but this comes at the cost of $i$'s signal no longer being aggregated into the election outcome. In order to elucidate how this tradeoff affects information aggregation, we focus on pure common-interest committees.

\subsection{Common-interest committees and neutral equilibria} 
\label{sec: LDex1}
In this section we  consider environments which are symmetric with respect to the state: there are no partisans ($n_A=n_B=0$), and the prior is $\pi=1/2$.\footnote{All results also apply if $\pi$ is close to $1/2$ or if the number of $A$ and $B$ partisans is equal, or if there are sufficiently many uninformed voters who can neutralise the partisans.} In this symmetric environment with no partisans, the best equilibrium is the one that best aggregates voters' information. Consider the best equilibrium of LD in the setting of Proposition 2; according to Proposition 2, this best equilibrium is neutral with respect to state in the sense that every voter's strategy is symmetric with respect to their signal realization. Formally, we say a strategy profile  $\sigma$ under LD is \emph{neutral} if under $\sigma$, every voter $i$ employs a neutral strategy by either: voting sincerely, abstaining at both signal realizations, or delegating to the same voter at both signal realizations (if $\sigma_i(a) = d_j$ then $\sigma_i(b)=d_j$). For DD, neutrality is defined similarly without allowing for delegation. In RD, a neutral strategy for a representative $i \in \mathcal{J}$ is either voting sincerely or abstaining at both signal realizations; a neutral strategy for a nonrepresentative $i \not \in \mathcal{J}$ is to either abstain at both signal realizations, delegate to a given independent partisan at both signal realizations, or to delegate to $A$ partisan $a^*$ at signal realization $s_i=a$ and to $B$ partisan $b^*$ at signal realization $s_i=b$ (this last type of strategy is analogous to sincere voting under DD and LD).  

Surprisingly, even in symmetric environments, it need not be the case that the best equilibria of LD are neutral. We present an example in the Appendix in which all best equilibria involve voters employing asymmetric strategies and electing the correct alternative with higher probability in one state than the other. These equilibria can rely on voters conditioning delegation on their signal realization, which allows for correlating voting between multiple voters.\footnote{If voter $i$ delegates to $j$ only when her signal is $a$, then the posterior probability of state A is higher conditional on $j$ casting an additional vote.} While delegation results in the loss of some private information of the delegator, conditional delegation can be used to mitigate this. However, these strategies require considerable coordination between voters and characterizing the best equilibrium is difficult. Instead, we focus on the best neutral equilibrium, which is a natural benchmark in symmetric settings. We  show that even when restricting to this natural class of equilibria -- hence shutting down conditional delegation -- voters can still use delegation in LD to do strictly better than under DD and RD.

Lemma \ref{lemma_best_neutral} in the Appendix shows that, in committees with no partisans and a flat prior, the best neutral strategy profile under all three mechanisms is an equilibrium. This is because when the environment is symmetric, a voter always has a neutral best response when all others use neutral strategies.\footnote{In fact, any neutral strategy profile at which there are no profitable deviations using only neutral strategies will be an equilibrium.} As strategy spaces are finite, every symmetric game hence has a neutral equilibrium in LD,  RD, and DD. In the remainder of the section we analyze some natural classes of committees. We demonstrate that the best neutral equilibrium of LD can do strictly better than that of DD and RD and also discuss `bad' neutral equilibria under LD.

\textbf{The single expert committee.} Here we consider the simplest setting in which liquid democracy may be useful in improving information aggregation. The committee is composed of $N-1$ `nonexperts' $1,...,N-1$ and a single `expert' $e$. Assume $N$ is odd.\footnote{Results are the same for $N$ even, but with $N$ odd ties are avoided which simplifies exposition.} All committee members $i \in \mathcal{N}$ are independents: $p_i = I$. The nonexperts $i=1,...,N-1$  receive private signals with precision $q_i =q\in (0.5,1)$ and the expert receives information with superior precision $q_e = r \in (q,1)$. Additionally, suppose that $\pi=\frac{1}{2}$ so the information setting is entirely symmetric.

This is a common-interest setting where all voters are purely interested in information aggregation. Intuitively, when the expert is sufficiently well informed relative to the nonexperts ($r$ is high relative to $q$), there is an incentive for the nonexperts to delegate to the expert. Delegation however comes at a cost: if a voter $i$ delegates, she cannot express her own private information in voting. We are interested in how voters solve this tradeoff in delegation under the best neutral equilibrium and other neutral equilibria.

Let $w^*(e) = \frac{\log(\frac{r}{1-r})}{\log(\frac{q}{1-q})}$; $w^*(e)$ can be thought of as the number of non-expert signals that  collectively have  informational content equivalent to the expert's private signal. By \citet{nitzan1982}, if the expert had $w^*(e)$ votes and the nonexperts had $1$ vote each, all voters voting sincerely would efficiently aggregate the information of the committee. Proposition \ref{prop_single_exp} characterizes the best neutral equilibrium of the single-expert committee under LD as a function of $w^*(e)$.

\begin{prop} \label{prop_single_exp}
There is a best neutral equilibrium $\sigma$ of the single expert committee characterized as follows:
\begin{enumerate}
\item If $w^*(e)<2$: all voters vote sincerely under $\sigma$.
\item If $w^*(e) \geq 2$: Exactly $\min\{\floor{w^*(e)}-1,\frac{N-1}{2}\}$ nonexperts delegate to the expert (i.e. $\sigma_i(a)=\sigma_i(b)=d_e$). The remaining voters vote sincerely. 
\end{enumerate}
These strategies are independent of the number of nonexperts in the committee. Whenever $w^*(e) \in (2,\frac{N+1}{2}]$, the best neutral equilibrium of LD does strictly better than that of DD. 
\end{prop}

The proof is in the Appendix but we explain the intuition. Note that if $w^*(e)<2$, then the best neutral equilibrium in liquid democracy, as per the Proposition, requires no delegation and hence is replicable under direct democracy. Meanwhile if $w^*(e) \geq \frac{N+1}{2}$, expert receives $\frac{N-1}{2}$ delegated votes; as she votes sincerely, she then dictates the election outcome. This outcome is also replicable under direct democracy by all nonexperts abstaining. Hence, liquid democracy strictly improves on direct democracy when the expert sufficiently well informed relative to nonexperts, but not too well informed. A comparison between liquid and representantive democracies will depend on the set of representatives $\mathcal{J}$, as discussed at the beginning of Section \ref{sec: pub}. For instance, the best neutral equilibrium of liquid democracy will do strictly better than that of representative democracy whenever $w^*(e)>2$ and $e \not \in \mathcal{J}$ (as the expert's superior information cannot then be given additional weight via delegation), or whenever $w^*(e) \in (2, \frac{N+1}{2}]$ and $\mathcal{J}=\mathcal{N}$ (in which case representative and direct democracy are equivalent).

We now give intuition for the structure of the best neutral equilibrium of LD. First note that if all nonexperts had a single vote and the expert held $\floor{w^*(e)}$ votes, all voters voting sincerely would produce equivalent election outcomes to the case where the expert had exactly $w^*(e)$ votes.\footnote{This is simply because with an odd number of voters and no abstentions under the strategy profile, all decisions are made by a margin of at least $1$ vote; hence rounding the expert's weight will not swing the election outcome.} In Case 1, when $w^*(e)<2$, $\floor{w^*(e)}=1$ and so all voters voting sincerely will attain the first-best outcome and hence must be the best neutral equilibrium. In this case, the best neutral equilibria of LD and DD perform identically.

Case 2 is the more interesting case:  The prescribed strategy profile allocates $\min\{\floor{w^*},\frac{N-1}{2}+1\}$ votes to the expert. When $w^*(e) \geq \frac{N-1}{2}+1$, the expert's signal contains more information than any $\frac{N-1}{2}+1$ other voters and this implies it is optimal for the expert to dictate the election outcome. More interestingly, when $2 \leq w^*(e) < \frac{N-1}{2}+1$, then delegation to the expert such that she attains a weight of $\floor{w^*}$ is optimal. As all nondelegating nonexperts are voting sincerely, this means the committee efficiently aggregates the information of the expert and all $(N-1) - (\floor{w^*}-1)$ nonexperts holding votes. However, this comes at the cost of wasting the private information of $(\floor{w^*}-1)$ delegators, meaning all information of the committee is not being efficiently aggregated. The implication is that when the expert is sufficiently, but not too well, informed it is worth wasting some voters' information in order to efficiently aggregate the information of the remaining voters. The best neutral equilibrium forms a `subcommittee' of decision makers who endogenously attain optimal weights.


To understand why it is optimal to delegate exactly $\floor{w^*}-1$ votes to the expert when $2 \leq w^*(e) < \frac{N-1}{2}+1$, suppose the expert is being delegated $k$ votes and consider the incentive of another nonexpert $i$ to delegate to her. For simplicity, assume that all voters who are not delegating vote sincerely (we show in the Appendix that this is the behavior at the best neutral equilibrium). Voter $i$'s delegation decision will only affect payoffs when: (1) $i$'s vote is pivotal, and (2) $s_i \neq s_e$ (if $i$ and $e$ receive the same signal realization, as nondelegators are casting their votes sincerely, delegation will not change how this vote is cast). When both these conditions hold, it must be that $s_i$'s realization agrees with that of $\frac{N-1}{2}$ other nonexperts; meanwhile $s_e \neq s_i$ and $s_e$ agrees with $\frac{N-1}{2}-(k+1)$ nonexperts. Hence, in the event that $i$ delegating to $e$ is pivotal, $i$'s private information agrees with that of $k+1$ \emph{excess} nonexperts while disagreeing with the expert's. Voter $i$ should delegate to the expert if and only if the expert is more likely to be correct. This is true when the informational content of $s_e$ is worth that of at least $k+2$ nonexpert signals; i.e. if and only if $k+2 \leq w^*(e) \iff k+1 \leq \floor{w^*}-1$.

There are a few important implications of this result. First, note that the number of votes the expert should be delegated in the best neutral equilibrium does not depend on the number of nonexperts in the committee. This is because the expert's optimal weight does not depend on this size of the committee, but only on the quality of her information relative to that of a nonexpert. Second, whenever $2 \leq w^*(e)<\frac{N-1}{2} + 1$, the best equilibrium involves nonexperts employing asymmetric strategies: either some nonexperts delegate while others vote sincerely. This means that there will be multiple best equilibria as nonexperts are interchangeable. Coordinating on equilibria in which voters of the same type use different strategies may be hard; we discuss this in more detail in Section 5. The following example demonstrates that while the best neutral equilibrium given in Proposition 2 highlights that LD can do strictly better than DD, miscoordination in delegation can lead to worse neutral equilibrium outcomes as well.

\begin{example}
\label{ex:overdelegation}
 Consider a single expert committee with  nonexperts $i=1,...,8$ (so $N=9$)   having precisions $q_i=0.6$ and the expert having precision $r=0.7$. There is a (neutral) equilibrium in which: \begin{itemize}
\item Nonexperts $i=1,2$ delegate to the expert ($\sigma_i(a)=\sigma_i(b)=d_e$) at both signal realizations.
\item Nonexperts $i=3,5,7$ delegate to nonexpert $i+1$: $\sigma_i(a)=\sigma_i(b)=d_{i+1}$ at both signal realizations.
\item Nonexperts $i=4,6,8$ vote sincerely.
\end{itemize}
\end{example}
When $q=0.6$ and $r=0.7$, $w^*(e) \in (2,3)$ and $\floor{w^*(e)} = \floor{\frac{\log(0.7/0.3)}{\log(0.6/0.4)}}=2$.  This means the expert should optimally vote with twice the weight of nonexperts. The best neutral equilibrium of Ls has one nonexpert delegating to the expert and the rest voting sincerely. In the equilibrium above, there is `overdelegation' to the expert, who ends up with $3$ votes. This overdelegation would make the expert suboptimally too powerful relative to nonexperts unless nonexperts delegate to one another in order to achieve the correct balance of power amongst those casting votes. Notice that the strategy profile given above leaves 3 nonexperts ($4,6,8$) each with $2$ votes. The alternative the expert votes for will win the election unless all $3$ nonexperts receive information disagreeing with her; as $w^*(e) \in (2,3)$, this means that this equilibrium first-best aggregates the information of the expert and three nonexperts holding votes (the correct decision is always made when only taking into account these four voters' information). However, obviously this equilibrium is inefficient as nonexperts signals are wasted due to overdelegation.

The single expert committee demonstrates the role delegation can play in allowing voters to optimally balance power relative to how informed voters are. However even in such simple settings and selecting only for neutral equilibria, mis-coordination in delegation can create incentives to waste further information through delegation. These incentives to agglomerate votes can result in inefficiency in equilibrium. Inefficiencies due to wasted delegation stemming from overdelegation to a small number of experts have been noted in the Computer Science literature on LD (\cite{kahng2018},\cite{halpern2021defense}), when testing liquid democracy experimentally (\cite{Campbell2022}), and in practice (see footnote \ref{new_scientist}) but to our knowledge have not been shown to be possible in equilibrium when voters are strategic.

\textbf{Committees with nonexperts and few experts.} We now show that the intuition from the best neutral LD equilibrium of the single expert committee extends to another natural class of committees. Consider a committee in which all members are independents and the prior is $\pi = \frac{1}{2}$. Suppose there are $n$ (even) voters who each have precision $q_i=q$; call these voters the `nonexperts'. The remaining $E=N-n$ (odd) voters are `experts'; each such voter $i$ has precision $q_i>q$ and $w^*(i) = \frac{\log(\frac{q_i}{1-q_i})}{\log(\frac{q}{1-q})}$.\footnote{Assuming $n$ even and $E$ odd is for convenience.} This committee is an adaptation of the extension of the single expert committee allowing for heterogeneity amongst multiple experts. When the number of nonexperts is sufficiently large (i.e. the number of experts is relatively few), the results from the single expert committee extend intuitively.

\begin{prop}\label{prop_multiexp}
There exists $\bar{n}  \in \mathcal{N}$ s.t. for committee described above with $n \geq \bar{n}$, there is a best neutral equilibrium of LD $\sigma$ in which: \begin{enumerate}
\item For each expert $i$, exactly $\floor{w^*(i)}-1$ nonexperts delegate to her (i.e. for  $\floor{w^*(i)}-1$ nonexperts $j$, $\sigma_j(a)=\sigma_j(b)=d_i$). 
\item All nondelegators vote sincerely.

\end{enumerate}

Moreover the best neutral equilibrium in LD is strictly better than the best neutral equilibrium in DD whenever $w^*(i)>2$ for some expert $i$.
\end{prop}

Proposition \ref{prop_multiexp} shows that the logic of Proposition \ref{prop_single_exp} extends to other settings. As long as the number of worst informed voters is large enough, optimal information aggregation---when restricted to neutral strategy profiles---in common-interest committees requires delegation such that all nondelegators attain (approximately) their  \citet{nitzan1982} optimal weight. Again, the committee forms a subcommittee of approximately optimally weighted voters to make the decision. Here as in the single expert committee, the optimal number of delegated votes each expert should receive does not depend on the size of the committee. Also, the best neutral equilibria  typically requires nonexperts to employ asymmetric strategies (some delegate and some vote sincerely) and hence there may be multiple such equilibria. 

With multiple experts in the committee, we cannot directly extend the single expert committee result -- that experts attain their rounded optimal weights in the best neutral equilibrium -- without there being a sufficiently large number of nonexperts. When the number of nonexperts is too small, the best neutral equilibrium may not have this property. This stems from the fact that efficient information aggregation requires every two voters to have relative weights that correctly reflect their relative information qualities, but this can be hard to achieve in small committees because delegation causes discrete changes in voters' weights. When the number of nonexperts is sufficiently large, however, we can extend the logic of Proposition \ref{prop_single_exp}. Importantly, Proposition \ref{prop_multiexp} is not an asymptotic result; it characterizes the best neutral equilibrium for finite (but large enough) $n$.

The proof is in the Appendix, but the logic is as follows. First, we show that as the committee gets large, the best neutral equilibrium must have a bounded number of voters delegating. This is because if the number of delegators grew arbitrarily as the committee grew, the committee would be wasting an unbounded amount of information (that of the delegators), which  is inefficient. As the number of nonexperts grows, the best neutral equilibrium must then require that almost all of these nonexperts do not delegate. Then, voting nonexperts dominate in the committee and what is most important for information aggregation is that the experts are optimally weighted relative to them. Hence the same incentives for nonexpert-expert delegation as in the single expert committee prevail and we attain a similar result. Moreover, whenever some expert has an optimal weight above $2$, the best neutral equilibrium of DD does strictly worse as this expert cannot attain her optimal weight without delegation. As with the single expert committee, the comparison with the best neutral equilibrium of RD depends on the set of representatives, but LD will do strictly better whenever, for instance, some expert $i$ with $w^*(i)>2$ is not a representative (preventing her from being delegated more votes), or if all voters in the committee are representatives (rendering RD and DD identical).

Numerous inefficient neutral equilibria will generally exist in these committees  under LD. As in the single expert committee, such equilibria can for instance be generated by overdelegation to experts and delegation between nonexperts to compensate.

\subsection{Inefficiency in equilibrium}
In the following we  denote the number of independent, $A/B$ partisan and uninformed independent voters as $n_I,n_A,n_B$ and $n_U$, respectively. 
As discussed before, it is hard to characterize the best equilibrium for general committees, let alone the full set of equilibria. However, it is easy to see that bad equilibria always exist. For instance, in LD, RD  and DD, it is a BNE for all voters to vote unresponsively for option $a$; as no voter is ever pivotal under this strategy profile, everyone is totally indifferent across all strategies.\footnote{Note that voting unresponsively for $a$ may be weakly dominated for some voters. However, it is not hard to come up with committees for which all independents voting for $a$ is an undominated equilibrium behavior and results in $A$ being the outcome of the election w.p. $1$.} This sort of equilibrium multiplicity is common in voting games and the problem is exacerbated when delegation is allowed in LD. For example, in any committee if we fix an independent voter $i$, whenever the number of independent voters $ n_I> |n_A-n_B|+1$, all independents delegating to $i$ (and $i$ voting sincerely) can be supported in equilibrium. Under such a strategy profile, $i$ dictates the election outcome and no individual voter can change this. Considering again the single expert committee, this means that it is an equilibrium for all voters to delegate to any particular nonexpert $i$, who votes sincerely. Hence equilibrium can support even the worst informed voter dictating the election outcome. Importantly, both of the types of equilibria just discussed can result in the committee failing to improve outcomes even as the committee (and hence the sum of its' members information) gets arbitrarily large.  This means equilibria in under all voting systems can fail to aggregate information asymptotically, something that, for instance, the best equilibrium of DD is capable of doing (e.g. sincere voting by all voters does so by the Condorcet Jury Theorem). Further, the former equilibrium example (of everyone voting unresponsively for one alternative) will typically be strictly worse than the best equilibria of direct democracy and representative democracy, hence the welfare comparison between liquid democracy and these other mechanisms is not clear-cut. 

Past literature has focused on undominated strategies as a selection criterion to deal with equilibrium multiplicity (\citealp{swing1996,FeddersenPesen1997}). We note that, even when making this selection, typical committees still admit inefficient equilibria of the types mentioned in the previous paragraph under all three mechanisms. However, in the next section we show that, for certain types of committees, applying \emph{iterated} deletion of weakly dominated strategies can be very powerful under LD.

\section{Dominance solvability with complete information} \label{sec: dominance}

While we have shown that when there is complete information about types, liquid democracy admits better equilibrium than mechanisms without delegation, LD also has numerous bad equilibria. Such equilibria can occur due to overdelegation or mis-coordination between voters in delegation. As the best equilibrium typically involves agents of the same type employing asymmetric strategies (e.g. some delegating and some not delegating), miscoordination is a concern. In Section \ref{subsec: deleg} we discuss one way in which this concern is mitigated: Proposition \ref{prop: dominant} implies that when some independent member of the committee is fully informed, then the game with delegation is dominance solvable (DS) and the DS solution is the best equilibrium. In such situations, it may be reasonable to expect voters to coordinate in playing the best equilibrium. Meanwhile, the game without delegation is not necessarily DS suggesting that, in certain settings, liquid democracy may make it easier to coordinate on the efficient equilibrium.

 We have already argued that it is weakly dominant for partisan voters (or partisan representatives in RD) to always vote in favor of their preferred alternative.  Assume that  $n_I \geq |n_A-n_B|$. 
 Suppose that some independent member of the committee $e \in \mathcal{N}$ has precision $q_e = 1$; that is, at least one member of the committee learns the state with certainty. It is straightforward to see that it is weakly dominant for $e$ to vote sincerely; conditional on her pivotality, sincere voting is uniquely optimal (as she knows $\omega$ and wants the election outcome to match it) and conditional on her not being pivotal, she is indifferent across all strategies. 

Below we provide an example showing that the game without delegation (direct democracy) is not DS while the game with delegation (liquid democracy) is. \begin{example}

Suppose there are 3 $A$-partisans, 1 $B$-partisan, 1 expert independent ($q_e=1$)  and  4 non expert independents ($q_i<1$).  It is easy to see that the partisans and expert voters have dominant strategies: an $A$ ($B$ partisan , expert independent) is pivotal if excluding her, the remaining votes are divided equally between $a$ and $b$. In all other cases they are indifferent between their top ranked candidate and others. Thus, assume that partisan and expert independents are voting their dominant strategies and consider the reduced game and the strategies of non-experts.

 The efficient equilibrium is chosen when  non expert independents  neutralise the partisans - so 2  vote $b$ unresponsively and two abstain unresponsively, ending up with a tie (excluding the expert). This ensures that the expert  decides the election. However this strategy requires a lot of coordination among the non expert independent voters. Moreover, in this example, each of the non experts is pivotal so the strategies of voting $b$ unresponsively and voting $x$ unresponsively are a unique best response.  The strategy of voting $a$ unresponsively is a unique best response 
 if 3 non experts vote $b$ unresponsively while one votes $a$ unresponsively. (This requires at least $n_e+ |n_A-n_B|+1$ non experts.)  This ensures that even if the partisans and expert voters have weakly dominant strategies, the reduced game does not allow further elimination of dominated strategies and is not dominance solvable. 
 
 With delegation, each of the non experts can delegate to the expert and achieve the efficient outcome.  Moreover, delegation weakly dominates all other strategies: consider the profile where 3 of the non experts vote $b$ unresponsively, then there is a tie in state $a$ excluding one non expert. Delegation ensures the correct outcome in both states while voting $b$ unresponsively delivers $b$ in state $a$. Voting $x$ unresponsively delivers a tie in state $a$. Thus delegation is strictly better than these two strategies. Consider the profile where 1 of the non experts votes $b$ unresponsively while 2 abstain, then the remaining non expert is pivotal in state $b$.  Delegation achieves the correct outcome in both states, while voting $a$ unresponsively delivers $a$ in state $b$.  Therefore delegation is a weakly dominant strategy and the game is dominance solvable.

\end{example}
Let $n_e$ be the number of perfectly informed independent voters. Applying iterated deletion of weakly dominated strategies, assume all partisans and $n_e$ perfectly informed independents vote as prescribed and consider the behavior of independents $i$ with $q_i<1$. We can show that as long as these nonperfectly informed independent voters are pivotal with positive probability in both states, the game is DS because it is iteratively weakly dominant for each independents with $q_i<1$ to delegate to some voter $e$ with $q_e=1$. This is simply because delegating to a perfectly informed voter can never hurt (as the perfectly informed voter always makes the correct voting decision) and will be strictly better for some strategies of others'. As $n_I \geq |n_A-n_B|$ and $n_e \geq 1$, under the DS solution the election outcome is determined by the perfectly informed independents and hence will always match the state of nature.\footnote{Strictly speaking, when $n_I = |n_A-n_B|$, the election outcome will be correct in one state and a tie \textemdash and hence correct w.p. $0.5$ \textemdash in the other state. This will still be the best equilibrium.} Note that if voters are able to understand that perfectly informed voters have an undominated strategy, the DS solution should be easy to coordinate on. The solution is also `type-symmetric' in a sense that all nonperfectly informed independents can delegate to any perfectly informed voter. All perfectly informed voters adopt the same strategy as do all partisan voters of each type. One can argue that reaching this DS solution has low cognitive cost. After presenting Proposition \ref{prop: dominant}, we argue that the same argument goes through as long as some independent voter in the electorate is sufficiently well informed, even if no one is perfectly informed.


In contrast, the game without delegation is not DS (Proposition \ref{prop: nodeleg}). As a result, reaching the efficient pure strategy Nash equilibria (PSNE) is not a straightforward task: it requires a lot of  coordination among the nonperfectly informed independent voters as efficient equilibria are in type-asymmetric strategies. Moreover, other types of equilibria that require less coordination (type-symmetric PSNE or mixed strategy Nash equilibria) exist, but they  frequently lead to inefficient outcomes.

Now we extend this example to the case of RD. The results are easy to show when the set of delegates includes the expert, in addition to one of each partisan. As before, each of the delegates is pivotal when votes are equally split among the two other voters. Therefore voting sincerely is a dominant strategy in the second stage.

The remaining voters are 2 $A$ partisans, and 4 non expert independents. 
Consider the reduced game where  delegates use their dominant strategies. In the first stage, therefore it is easy to show that voting for delegate  $A$ is a dominant strategy for the $A$ partisans and voting for the expert is a dominant strategy for the non experts.\footnote{Since there are 6 voters, we can create a profile where  a tie is created in favour of the favoured delegate. }. In this case, RD reproduces what LD does endogenously and the outcomes are the same. 

Consider a different set of delegates which includes a non-expert independent voter. The expert voter has a dominant strategy of voting for partisan delegate $A$ when the state is $a$ and delegate $B$ when the state is $b$. The non expert delegate must coordinate with the non expert voters. To achieve the correct outcome for sure, two of the non expert delegates must vote to neutralise the $A$ partisans, therefore two of the nonexperts vote for $B$ unresponsively while two nonexperts (including the delegate) abstain. Moreover each of these strategies (voting $b$ unresponsively or abstaining unresponsively) is a unique best response for the nonexperts.   Since nonexperts are identical, this implies that these strategies cannot be eliminated for any of them.  Suppose, instead, in the reduced game,  two of the nonexperts vote for partisan $B$, one of them votes for $A$ and the remaining one abstains. In this case, the expert decides the election but each of the strategies of the nonexperts in a unique best response, therefore the game is not DS. 
Overall, therefore, there exists a committee where the best RD equilibrium is  equivalent to the best equilibrium in the  LD game and the game is DS  but the expert is chosen exogenously to be in the committee, and there exists a committee where  the best RD equilibria are equivalent to the best DD equilibria, and the game is not DS. In this case, in order to achieve  full information equivalence, the non-experts must have asymmetric strategies. 

In Section \ref{sec: LDex1}, 
we showed that the best equilibria in liquid democracy can involve asymmetric voting/delegating amongst voters of the same type. In the LD cases discussed, one voter's decision to delegate may affect another's incentive to do so. For example, in the single expert committee, nonexpert $i$ may want to delegate to the expert if and only if sufficiently few other nonexperts have delegated. The argument presented in this section shows that issues of coordination on the best equilibria are less of an issue when a voter's delegation incentives do not depend on others' behavior: when some voter is sufficiently well informed, others delegating to her is good regardless of what other voters are doing and hence the best equilibrium can be easy to coordinate on. 

In summary, there are two reasons why we might want the game to be DS. First, because,  as we see in the example, this property rules out many other equilibria which are inefficient and implausible in the game with delegation. Second, without this property, reaching the efficient equilibria requires a fine balancing act, where independent voters must use type-asymmetric strategies. We have seen in the previous section how coordination mistakes can lead to far from efficient informational outcomes. When the game is DS, however, liquid democracy allows for easier coordination on the best equilibrium in an unambiguous sense.

\subsection{With delegation}
\label{subsec: deleg}

Proposition \ref{prop: dominant} below shows that the election game with delegation is dominance solvable, and characterizes the DS outcome.
Recall that $n_U$ is the number of independents who have type $(I,q_i)$, where $q_i<1$.
 
 \begin{prop}
 \label{prop: dominant}
 Suppose that $n_e \geq 1$, $n_U \geq n_e+ |n_A-n_B|+1$ 
 and $max(n_e+n_A,n_e+n_B) \leq \frac{N}{2}$.
Then  the election game with delegation is dominance solvable. The DS outcome is the efficient outcome. 
\end{prop}

The proof of this proposition is in the Appendix. 

The dominant strategies for $A$, $B$ and independent and perfectly informed voters are $\sigma_i(a)=\sigma_i(b)=a$, $\sigma_i(a)=\sigma_i(b)=b$, and $\sigma_i(a)=a,\sigma_i(b)=b$ respectively, while the strategy for an independent voter $i$ who is not perfectly informed in the DS equilibrium is $\sigma_i(a)=\sigma_i(b)=d_j$ for some perfectly informed $j$. Given the conditions on numbers of each type, this profile ensures that alternative $a$ wins in state $\omega=a$ and alternative $b$ wins in state $\omega=b$. The inequality  $n_I \geq |n_A-n_B|$ ensures independents are pivotal with positive probability which ensures sincere voting is a weakly dominant strategy for perfectly informed voters ($n_e \geq 1$). Finally,  $max(n_e+n_A,n_e+n_B) \leq \frac{N}{2}$ is needed to guarantee each nonperfectly informed independent $i$ is pivotal under some profile in each state. Given this, $i$ delegating to a perfectly informed voter is weakly dominant as it is strictly preferable to all other strategies when $i$ is pivotal (and weakly preferable otherwise). 

While the assumption that some voter in the committee is perfectly informed ($q_i=1$) is a strong one, this is not a necessary condition for Proposition \ref{prop: dominant} to hold.  It would be sufficient, for instance, for there to be some independent member of the committee $e$ whose private signal is more informative than all other independents' private signals combined. In other words, $$w^*(e) > \sum_{\{i \in \mathcal{N} | i \neq e, p_i=I\}} w^*(i).$$ Again, it is weakly dominant for partisans to vote for their preferred outcome. Then, if $n_I \geq |n_A-n_B|$, it is a weakly dominant strategy for $e$ to vote sincerely. This is because, any updating she would do when conditioning on her own pivotality is based on the information held by other independents (as partisans' behavior is deterministic). If her own signal is relatively informative enough, this updating will be outweighed by the updating from her private signal. Note that under this condition, a best equilibrium (in terms of information aggregation) is $e$ dictating. Again, an independent voter $i$, such as $i \neq e$ cannot be harmed by delegating to $e$ and  strictly benefits when pivotal with positive probability.
Note that in the case of RD, the game is DS  by  similar arguments iff at least one sufficiently well informed expert is in $\cal J.$ But the latter is not guaranteed, so RD could be equivalent to DD or to LD depending on the independent experts in the set $J$.

Next, we move to the game without delegation and show that it is not DS.

\subsection{Without delegation}

\label{subsec: nodeleg}
\begin{prop}
\label{prop: nodeleg}
 Suppose that $n_e \geq 1$, $n_I \geq |n_A-n_B|$, $max(n_e+n_A,n_e+n_B) \leq \frac{N}{2}$, and that delegation is not allowed. Then: \begin{enumerate}
 \item  The game is not dominance solvable. 
 \item There exist multiple efficient equilibria in asymmetric strategies for nonperfectly informed voters where the outcome is $a$ in state $\omega=a$ and $b$ in state $\omega=b$. Assume $(n_I-n_e)\geq n_I+ |n_A-n_B|+1$.  If $n_e> |n_A-n_B|,$ there exists an efficient equilibrium where all nonperfectly informed independent voters abstain. If $n_e < |n_A-n_B|$ then there exists an inefficient equilibrium where all nonperfectly informed independent voters abstain and the correct outcome is not chosen in at least one state. There also exist inefficient symmetric PSNE where the outcome is $i \in \{a,b\}$ regardless of the state.  
 \end{enumerate}
\end{prop}

The proof of this proposition can also be found in the Appendix.

Overall, we observe that the efficient outcome is more "likely" (in the sense of DS of the game) to be reached in the game with delegation than in the game without delegation. The reason is that it is much easier for voters to coordinate their strategies in the game with delegation as the DS equilibrium strategies are type-symmetric. In contrast, if the game is not DS, the efficient equilibrium requires high level of coordination as there are multiple such equilibria (any permutation of an equilibrium strategy profile across nonperfectly informed independent voters is an equilibrium). Moreover there are also multiple inefficient equilibria even if we restrict ourselves to PSNE. All outcomes are possible in this game. 

Unlike much of the literature on information aggregation  (e.g. \citealp{swing1996}), we do not need to assume large elections for this result. With delegation, when types are known  e.g. in committee elections, delegation ensures that the efficient outcome is reached. When delegation of votes is not allowed it is "likely" that inefficient outcomes prevail. 
Of course the more plausible case is of incomplete information on types. Section \ref{sec: private} describes this. 

\section{Incomplete information about types} \label{sec: private}
 
In this section we consider the type of each agent as her private
information. This is an important generalisation as preferences and
information precision cannot be publicly observed in reality unless a truth
revelation mechanism is in place. We still allow each agent to be
characterised by a different likelihood of being assigned each specific
type, which makes the model more realistic compared to the standard
assumption that all agents draw their type from a common distribution.

A first challenge that we have to face is with respect to equilibrium
existence. Indeed, if the space of types was finite, existence of a
(Bayesian Nash) equilibrium would follow trivially from standard equilibrium
existence arguments for finite games. Since we want to be more general and
permit atomless distributions, we have to provide independent arguments that
an equilibrium exists. Throughout this section, we focus on pure undominated
equilibria in threshold strategies. That is, we do not constrain the
strategy set of the players, but only the set of equilibria that we are
interested in. Focusing on undominated strategies imposes that partisans do
not delegate and simply vote for their preferred party. Threshold
strategies imply that if two common value types of an agent characterized by different precisions take the same
action, then no common value type of this agent with an intermediate precision takes another action. Whenever we refer to an equilibrium from now on, we implicitly assume that it is
one characterized by the features mentioned.

The second challenge we  encounter relates to the welfare effects of
liquid democracy vis-a-vis direct and representative democracy. Since the
environment is much more complex than before, one cannot provide a
similarly detailed description of the set of equilibria in the general case.
What we can provide though is a comparative analysis of the best equilibria
of all systems in terms of information aggregation. This is compatible with
similar studies that also provide welfare comparisons based on the best
equilibria (given some welfare benchmark) of alternative collective choice
mechanisms (see, for instance, \cite{Ahn2016}). 

To establish our first result and to be able to proceed with the one
regarding welfare, we work in the following manner. First, we define an
auxiliary game that is identical to the one that we study except for a twist: in this auxiliary game, even types $A$ and $B$ have the same ex post utility
function as the common value voters. However, the action space of these
types is now constrained to only one action: type $A$ voters are only
allowed to vote for $A$, and type $B$ voters are only allowed to vote for $B$. It is straightforward that each equilibrium of the auxiliary game with
liquid, direct or representative democracy corresponds to a unique undominated equilibrium
of the game we are interested in with liquid, direct or representative
democracy, and vice versa. Therefore, we formally analyze the auxiliary
game, to understand whether the game of interest admits an equilibrium and
how the best equilibria of different kinds of democracy compare.\footnote{ Formally, a threshold strategy in the auxiliary game is an ex-ante strategy
and is characterized by a collection of thresholds $t_i$ $t_{i}\in \lbrack
0.5,1]^{n+1}$ for each player and signal realization, and a strict ordering $%
r_{i}$ of the $n+2$ available actions (i.e. the actions $a$, $b$, $x$, $%
d_{j\neq i}$) for each player and signal realization. When a player has a
type from $0.5$ up to the first threshold, she employs the action at the
top of the ordering, when her type is larger than the first threshold and at
least as large as the second threshold she employs the action that ranks
second in the ordering, and so on. A threshold strategy in the standard game
is identical to the one of the auxiliary game as far as independent types
are concerned, and employs the weakly dominant action for the partisan
types.}

Second, we observe that for any ex ante strategy profile, each player has a
best response in threshold strategies. That is, the profile of threshold
strategies that maximizes the ex ante common expected utility of the agents
must be an equilibrium of the auxiliary game (\cite{mclennan98}), and it is also
the best equilibrium in terms of information aggregation. Hence, our
equilibrium existence proof amounts to establishing that such a maximizer
profile of threshold strategies exists. We achieve this by arguing that the
ex ante common expected utility of the players in the auxiliary game is
continuous in threshold strategies, and that the set of
threshold strategies is compact.

\begin{lemma}
When types are the agents' private information,
an equilibrium always exists under liquid democracy.
\end{lemma} 


\begin{proof} As explained above, we focus on the auxiliary game and proceed in distinct
steps. First we show that for every strategy profile, there always exists a
best response that is a threshold strategy. Then, we prove that a
utilitarian planner who is constrained to choose a threshold strategy for
every player, can always solve her problem (i.e. a maximizer strategy
profile exists among the ones involving only threshold strategies). And
finally we combine \cite{mclennan98}, along with our first two steps, and show
that an equilibrium in threshold strategies always exists.

Notice that the interim expected utility of player $i$ in the auxiliary game
when she is of type $(I,q_{i})$, receives signal $a$, takes action $y\in
\{a,b,x,d_{j\neq i}\},$ and the other players are expected to use the
ex-ante strategy profile $\sum_{-i}$, which is given by
\begin{align*}
q_{i}\cdot \Pr(A \text{wins}|\omega =a,y,\sum_{-i})+(1-q_{i})  \cdot \Pr (B\text{ wins}|\omega=b,y,\sum_{-i}).
\end{align*}
This is a linear function of $q_{i}\in \lbrack 0.5,1]$.
Therefore, if we also take the expected utility of this player for any of
the other action $y^{\prime }$, the two of them---unless they fully
coincide---can intersect at most one $q_{i}$. That is, one can always
find a threshold value within $[0.5,1]$ such that using $y$ is optimal if we
are below (above) the threshold, and $y^{\prime }$ is optimal if we are
above (below) the threshold. This trivially generalizes to multiple actions,
and hence for any ex-ante strategy profile $\sum_{-i}$ of the other players,
there always exists a best response for player $i$ that is a threshold
strategy (even if $\sum_{-i}$ does not involve threshold strategies).

Now we move to the problem of a social planner who can fix the ex-ante
strategy of all the players (she is allowed to use only threshold
strategies), and wishes to maximize the sum of the players' utilities. Since
all the players have the same ex-post utility, the planner essentially
maximizes the probability that the outcome is correct (i.e. it matches the
state of the world). Given that the set of threshold strategies is compact,
and the probability  the correct outcome wins is continuous in the
threshold strategy employed for each player and signal realization (because
of the fact that precisions are drawn from an atomless distribution), it
follows that the planners problem has a solution. That is, there exists a
profile of threshold strategies that maximizes the ex-ante probability of
election of the alternative that matches the state of the world.

By \cite{mclennan98}, we know that the strategy profile that solves the
planner's problem is an equilibrium of the game in which players are
constrained to use only threshold strategies. While we do not constrain the
player to use threshold strategies, we have already established that among a
player's best responses there is always a threshold strategy. Hence, the
auxiliary game---and thereafter the game we are interested in---always
admits a pure undominated equilibrium in threshold strategies.\footnote{Notice that an ex-ante BNE in threshold strategies is
trivially an interim BNE of the game as well.} 
\end{proof}

Given the foundation laid above, and the arguments used to establish
existence of an equilibrium under liquid democracy, the main welfare results follow directly\footnote{Notice that the above proof can be adapted to also establish equilibrium existence under direct and representative democracy in this environment of incomplete information. Indeed, the only thing we need to change is the interim action set so that our argument corresponds to each specific form of democracy, and the rest follows.}.

\begin{thm}\label{main}
In liquid democracy, the best equilibrium (in terms of information aggregation) is at least as good as the best equilibrium in direct and representative democracy.
\end{thm}


In the auxiliary common-value game under direct/representative democracy, the threshold strategy profile that maximizes the probability of the correct alternative winning the election is the best equilibrium in terms of information aggregation. This strategy profile is also feasible in the auxiliary game under liquid democracy, and leads to weakly lower welfare compared to the threshold strategy profile that maximizes the probability of implementing the correct alternative---which constitutes an equilibrium. Since each equilibrium of an auxiliary game corresponds to a unique equilibrium of the corresponding game we are interested in, the result described in Theorem \ref{main} holds.

To show that the best equilibrium in liquid democracy can be strictly better
than the best equilibrium in direct and representative democracy, we consider the following example.



\textbf{An example: common-interest committee with experts and nonexperts and incomplete information on types.} Consider the following modification of the committee with nonexperts and multiple experts introduced in Section \ref{sec: pub}. As before, all $N$ committee members are independents, the prior is $\pi = \frac{1}{2}$, and there are $E$ (odd) experts with precisions $q_e=r$. Suppose that the $N-E$ (even) other voters, nonexperts $i=1,...,N-E$, each have precision drawn iid from a continuous distribution with support on $[0.5,\bar{q}]$ for some $\bar{q} \leq r$. While all voters are certain of the quality of the experts' information, and are certain this information is superior to that of other voters, the precision of each nonexpert's information is random and privately learned. This setup is studied in \cite{Campbell2022}; they focus on type-symmetric equilibria in which: (1) non-experts do not delegate to one another and all adopt the same ex-ante strategy, (2) all non-delegating voters vote sincerely. They show that such an equilibrium exists that strictly improves payoffs (and information aggregation) relative to simple majority voting. This equilibrium takes the form of some cutoff $\tilde{q} \in [0.5,\bar{q}]$ such that a voter $i$ delegates to an expert (picking one uniformly at random) if and only if $q_i < \tilde{q}$ and votes sincerely otherwise.\footnote{\cite{Campbell2022} does not allow for abstention, but the equilibria they find are still equilibria when allowing for abstention.} This type of equilibrium is very intuitive: voters delegate their vote if and only if the information they receive is of low enough quality. Interestingly, this equilibrium maximizes payoffs amongst all sincere voting  semi-symmetric equilibria: the equilibrium cutoff $\tilde{q}$ balances the costs and benefits of delegation optimally. For example (from \citet{Campbell2022}) when $N=3$, $K=1$, $r=0.7$, and $F(q)$ is uniform on $[0.5, 0.7],$ $\tilde q=0.572$ i.e. a nonexpert delegates if and only if their precision is below $0.572.$ This cutoff stems from an indifference condition for a nonexpert with exactly the cutoff precision. A nonexpert $i$'s decision to delegate their vote is conditioned on the event that their vote is pivotal, i.e. (i) $i$'s signal disagrees with the expert (ii) so does the other nonexpert's signal (iii) the other non expert is not delegating. In this case $\tilde q$ solves $$0.7\times(1-\tilde q)\times[1-\mu(\tilde q)]= 0.3\times\tilde q\times\mu(\tilde q),$$ where $\mu(\tilde q)$ is  the expected precision of the other non expert given her precision is above $\tilde q$. I.e. indifference at $\tilde q$ requires that the probability that two non experts are correct and the expert is wrong equals the probability that the expert is correct and both non experts are wrong. 


\textbf{Liquid democracy  with partisans.} As it has been recently shown, in settings like this one, one can improve
upon simple majority voting by introducing cumulative voting, according to which each agent is allocated a number of votes and she is allowed
to distribute them anyway she likes between the two alternatives (see \citealp{Bouton21}). Liquid democracy, like cumulative voting, allows agents to
better correlate their voting influence with their information precision,
leading also to better outcomes than simple majority voting. But the
underlying mechanisms behind these welfare improvements are not equivalent.
As the following example shows, liquid democracy can aggregate information
better in some cases than any voting rule that assigns the same voting power
to each agent, independently of whether one is allowed to abstain partially on some of her votes or not.

Consider that there are three players, and that the priors are even as in the previous example. The first player is a partisan but it is not known which party she supports (she is a partisan of each alternative with equal probability), the second and the third player are informed (i.e. their signal is always correct) independent voters with probability $1/2-x$, uninformed independent voters with probability $1/2-x$, and partisans of each party with probability $x$, for $x\in[0,1/2)$. Without delegation everybody votes sincerely (and the uninformed are indifferent between voting or abstaining), but the outcome is not always optimal (i.e. when 2 players are independent but only 1 is informed, the correct outcome does not win for sure). With delegation, though, the following equilibrium exists: All types except uninformed independent (henceforth, uninformed) ones vote sincerely, and the uninformed type of player two (three) delegates her vote to player three (two). Hence, whenever there are two independent voters and at least one of them is informed, the correct outcome wins for sure (which is strictly better compared to no delegation), and the correct outcome wins with probability $1/2$ otherwise (which is the same as with no delegation).

The reason why this happens is because the uninformed type of player two
understands that increasing the voting power of player three, can only affect
her expected utility if player three is informed. In that case, by delegating
her vote to player three she can guarantee that the correct outcome will
prevail, even when player one is a partisan supporting the incorrect
alternative. When player three is uninformed or partisan, then, from the point
of view of player two, the probability that the correct outcome wins is $1/2$ independently of her action. Therefore, with liquid democracy full
information equivalence obtains in this case. Notice that in direct
democracy \emph{under any anonymous voting rule}, there is no way that the correct
outcome wins with certainty when player two is
informed, player three is uninformed and player one is a partisan of the incorrect alternative, establishing our claim that vote-delegation is not equivalent to any anonymous voting system which precludes transfers of votes from one player to another.

Of course, liquid democracy---as modelled here---does
not always perform better than any kind of voting system. For instance, when
all voters are common value, but their information precision is their
private information, then direct democracy with cumulative voting can do
better than liquid democracy.\footnote{Consider the following simple version of cumulative voting: each player has one vote and can cast any fraction of it to any of the two alternatives. Under this formulation, there is always an equilibrium that achieves full-information equivalence (\citealp{Bouton21}).} Indeed, if we have three players and two of them are independent with information precision $q'$, while the third is an independent either with precision $q'$ or with precision one, then with cumulative voting full-information equivalence obtains in equilibrium, while with liquid democracy it does not.

These examples imply that the optimal way to go is by merging the two
interventions: i.e. allow agents to transfer votes both between alternatives
and voters. Indeed, one can trivially adapt the arguments supporting Theorem
2, to the case in which each agent is endowed not with a unique, but with
several votes instead, and establish that a hybrid system of liquid
democracy with cumulative voting supports an equilibrium that aggregates
information at least as well as the best equilibrium of the standard version
of liquid democracy analyzed here, and of direct democracy with cumulative
voting.

\section{Conclusion} \label{sec: conclusion}
We consider a model of liquid democracy focused on studying the information aggregation properties of this growingly popular voting system. We show that liquid democracy admits an equilibrium that outperforms, in terms of information aggregation and welfare, direct and representative democracies. Such an equilibrium exists even there are partisan voters who are not interested in information aggregation, and even when voters' types are not commonly known. When types are commonly known, we unpack what drives a voter's incentive to delegate and how voters optimally solve the tradeoff presented by delegation. We discuss when delegation is needed to strictly improve outcomes over other voting mechanisms, but also show that miscoordination in delegation can lead to bad outcomes. Finally, we argue that despite the presence of bad equilibria, in some settings liquid democracy may make it easier for voters to coordinate on efficient equilibria than other voting systems. 

Finally, one might wonder whether the specifics of the mechanism in place matter greatly in settings where communication is allowed. To be sure, in common value settings communication is known to improve welfare and to diminish differences between voting systems (see, for instance, \citealp{gerardi2007deliberative}). In settings with partisan voters, though, communication is not enough to guarantee an efficient outcome. Consider that after types and signals are drawn voters are asked to submit a report (e.g. a binary message), and then the distribution of reports is presented publicly. If the frequency of a given report affects monotonically the beliefs of the agents regarding the likelihood that state $a$ is correct, then such a report would also be chosen in equilibrium by partisan voters who support alternative $A$, making the distribution of reports less informative. 

Perhaps more importantly, even if there exists an efficient equilibrium with communication, it is never unique (i.e. an inefficient babbling equilibrium always exists) and it is never the case that the efficient outcome is implemented with iterated elimination of dominated strategies. Unlike delegation, where agents with high information precision have a weakly dominant strategy to vote for the alternative they think is most likely correct, when agents are asked to first communicate and then vote, sending a given message is optimal only conditional on how one believes that it will be perceived by others. That is, a perfectly informed common value voter, whose type is known to the other voters, will find it optimal to send message $x$ in the communication stage only if the other agents interpret $x$ in a certain way, and might not find it optimal otherwise. Therefore, while communication is a valid means to improve welfare, it does not seem to eliminate all the advantages of delegation, like the possibility of implementing the efficient outcome in certain settings without imposing demanding assumptions on beliefs. 

 \bibliographystyle{apalike}
\bibliography{vote}
\newpage
\appendix
\setcounter{table}{0}
\renewcommand{\thetable}{A\arabic{table}}
\setcounter{figure}{0}
\renewcommand{\thefigure}{A\arabic{figure}}

\section{Proofs}

\subsection{Proofs of Section 4.}
We first prove a Lemma.

\begin{lemma}
 For any independent voter $i$ in LD or DD, and any independent $i \in \mathcal{J}$ in RD, the following strategies are weakly dominated: \begin{itemize}
\item[(1)] $\sigma_i(a)=b$ and $\sigma_i(b) \neq b$
\item[(2)] $\sigma_i(a) \neq a$ and $\sigma_i(b)=a$.
\end{itemize}
\end{lemma}

\begin{proof}
We first prove the result under LD and then show that the same applies for DD and RD. We prove (1) and by symmetry (2) holds as well. Let payoff$(s_i,y,\sigma_{-i})$, $s_i \in \{a,b\}$, $y \in \{a,b,x,d_{j \neq i}\}$ be $i$'s payoff from taking action $y$ when her signal is $s_i$ and others are playing strategy profile $\sigma_{-i}$. $i$'s payoff from using strategy $\sigma_i(a)=b, \sigma_i(b)=y' \neq b$ is:

\begin{align*}
& Pr(s_i=a)payoff(a,b,\sigma_{-i}) + Pr(s_i=b)payoff(b,y',\sigma_{-i}) \\
\end{align*}    

As $i$ is an independent, we have: 

\begin{align*}
& payoff(a,b,\sigma_{-i}) \geq payoff(a,y',\sigma_{-i}) \\
& \iff q_i (Pr(\mathcal{O}=A|\omega=A,\sigma_{-i},y=b)-Pr(\mathcal{O}=A|\omega=A,\sigma_{-i},y=y')) \\
& + (1-q_i)(Pr(\mathcal{O}=B|\omega=B,\sigma_{-i},y=b) - Pr(\mathcal{O}=B|\omega=B,\sigma_{-i},y=y')) \geq 0 \\
& \iff |Pr(\mathcal{O}=B|\omega=B,\sigma_{-i},y=b) - Pr(\mathcal{O}=B|\omega=B,\sigma_{-i},y=y')| \\
& \geq \frac{q_i}{1-q_i}|Pr(\mathcal{O}=A|\omega=A,\sigma_{-i},y=b)-Pr(\mathcal{O}=A|\omega=A,\sigma_{-i},y=y')|
\end{align*}

where the last step comes from the fact that the first term on the LHS of the second equation is weakly negative; voting for $b$ when the state is $A$ results in a weakly lower probability of $A$ winning than taking any other action (this is because other player's actions are independent from $i$'s conditional on the state; note that this statement is true even if the other action we considering for $i$ is delegation to some $j$ and creating a cycle: then all the votes in the cycle would be abstained if $i$ delegates to $j$ but would be cast for $b$ if $i$ votes for $b$). Similarly, the second term on the LHS is weakly positive.

Meanwhile, similarly:

\begin{align*}
& payoff(b,y',\sigma_{-i}) \geq payoff(b,b,\sigma_{-i}) \\
& \iff q_i (Pr(\mathcal{O}=B|\omega=B,\sigma_{-i},y=y')-Pr(\mathcal{O}=B|\omega=B,\sigma_{-i},y=b)) \\
& + (1-q_i)(Pr(\mathcal{O}=A|\omega=A,\sigma_{-i},y=y') - Pr(\mathcal{O}=A|\omega=A,\sigma_{-i},y=b)) \geq 0 \\
& \iff |Pr(\mathcal{O}=A|\omega=A,\sigma_{-i},y=y') - Pr(\mathcal{O}=A|\omega=A,\sigma_{-i},y=b)| \\
& \geq \frac{q_i}{1-q_i}|Pr(\mathcal{O}=B|\omega=B,\sigma_{-i},y=y')-Pr(\mathcal{O}=B|\omega=B,\sigma_{-i},y=b)|
\end{align*}

As $q_i >0.5$, both of these two expressions cannot hold and so the strategy must be dominated by either $\sigma_i(a)=\sigma_i(b)=b$ or $\sigma_i(a)=\sigma_i(b)=y'$.

The same argument applies directly to DD and independents in $\mathcal{J}$ in RD, but without delegation in the strategy space. 

\end{proof}

Recall Proposition 2. We provide a proof below.

\textbf{Proposition 2.} Assume $\pi = \frac{1}{2},$ $n_A\geq n_B$, $n_U$ uninformed ($q_i=0.5$) independent voters and  $E$ independent experts ($q_i>0.5$). There exists $n^*$ such that when $n_U>n^*$, under the best LD equilibrium $\sigma$: \begin{enumerate}
\item $n_A-n_B$ uninformed independents vote unresponsively for $B$.
\item Each remaining uninformed independents either: (1) abstains at both signal realizations, or (2) delegates to the same expert at both signal realizations.
\item The election outcome under LD is the first-best outcome w.p. $1$.
\item Let $v_i,v_j$ be the number of votes held by experts $i$ and $j$. As $n_U \rightarrow \infty$, $v_i / v_j \rightarrow w^*(i)/w^*(j)$ (experts approach optimal weights).
\end{enumerate}
Before going on to the proof,
we define some useful notation. First, let $W^{*}(\mathcal{J}) = \sum_{i \in \mathcal{J}} w^*(i)$ be the summed \citet{nitzan1982} optimal weight of $\mathcal{J} \subset \mathcal{N}$. 

In the case of $\pi=\frac{1}{2}$, we define the following. Let $\mathcal{C}$ represent the set of all partitions of all independents into two subsets or coalitions. Here we represent profiles of signal realizations by the coalitions, $(C,C') \in \mathcal{C}$, they induce; $(C,C')$ denotes the event where all voters $i,j \in C$ receive signal realization $s_i = s_j$ ($=s_C$) and all voters in the complementary set, $C'$, get signal realization $s_{C'} \neq s_C$. Note that each $(C,C') \in \mathcal{C}$ occurs with positive probability.  Each $(C,C')$ also corresponds to two different signal realization profiles: one where voters in $C$ get signal $a$ and $C'$ get signal $b$ and vice versa; because $\pi=1/2$ and there are no partisans, these are informationally equivalent in that both events contain the same strength of evidence in favor of $\omega=s_C$ versus $\omega=s_{C'}$. We use the term `signal realization' to denote both a profile $(s_1,...,s_{I+1})$ and the coalitions $(C,C')$.


\textbf{Proof of Proposition \ref{prop_uninformed}.}

\begin{proof}

The best equilibrium implements the correct outcome with the highest probability. Hence any equilibrium which first-best aggregates the information of all experts must be the best equilibrium (uninformed independents have no information). It remains to show that such an equilibrium exists for $n_U$ large enough.

Assume $n_U \geq n_A-n_B$. For every expert $i \in E$ let $w^{*k}(i) = k w^*(i)$ for $k>0$. Suppose $n_A-n_B$ uninformed voters vote unresponsively for $B$. Fix $k>0$ and suppose that the remaining uninformed voters abstained and all experts $i$ had exactly weight $w^{*k}(i)$. As optimal weights can be scaled without changing election outcomes, this would first-best aggregate the information of experts. Consider any signal realizations partitioning the set of experts into subsets $C$ and $C'$ with disagreeing signals and, wlog,  $W^*(C) \geq W^*(C')$ (note the prior is flat, so it does not matter which signal each group received). Whenever $W^*(C)>W^*(C')$, the decision will be made by a margin of $k(W^*(C)-W^*(C'))$ votes. For large enough $k$ this quantity is greater than $E$ and hence the same decision would be implemented if we rounded down each $w^{*k}(i)$ to an integer (as rounding each $w^{*k}(i)$ to an integer will affect vote tallies by at most $E$). Whenever $W^*(C)=W^*(C')$, voters are indifferent between options under the first-best, and so rounding weights will leave payoffs first-best. 

Fix some such large $k$ and let the rounded weights be $\floor{w^{*k}(i)}$. Then for large enough $n_U$, the same outcomes can be replicated by an LD strategy profile where experts vote sincerely and amongst uninformed voters: $n_A-n_B$ vote unresponsively for $B$, $\floor{w^{*k}(i)}-1$ delegate at both signals to each expert $i \in E$, and the remaining always abstain.

For the final statement, as $k \rightarrow \infty$, $\floor{w^{*k}(i)}/\floor{w^{*k}(j)} \rightarrow w^*(i)/w^*(j)$.
\end{proof}

\subsubsection{Example of Best Non-Neutral Equilibria}

\textbf{Example 1} Suppose $\pi=0.5$ and consider a committee with: (1) one `expert' $e$ with precision $0.8$, (2) `nonexperts' $1,...,4$ with precisions $0.6$. All voters are independents.

The Nitzan and Paroush (1982) weight of the expert is $w^* \in (3,4)$ and that of the nonexperts is $1$. The best neutral equilibrium is as follows:

\begin{itemize}
\item Expert votes sincerely: $\sigma_e(a)=a$, $\sigma_e(b)=b$.
\item (At least) two experts delegate to the expert so she attains her rounded down Nitzan and Paroush (1982) weight ($\sigma_i(a)=\sigma_i(b)=d_e$ for at least two nonexperts $i$).
\item In this equilibrium, the expert dictates; this attains EU $0.8$ (expert's precision).
\item As $w^* \in (3,4)$, the expert dictating is good \emph{except} if all four nonexperts disagree with her. i.e. this equilibrium makes the ex-interim wrong decision only at signal realizations: $(a,a,a,a,b)$ and $(b,b,b,b,a)$ (with the last signal realization in the vectors corresponding to the expert's).
\end{itemize}

Now consider the best equilibrium overall. Note that if there is a delegation cycle ($i$ delegates to $j$ and $j$ delegates to $i$), our model assumes both votes are abstained (although similar examples can be constructed under alternative assumptions). The best equilibrium is:

\begin{itemize}
\item Expert votes sincerely: $\sigma_e(a)=a$, $\sigma_e(b)=b$.
\item Nonexperts $1$ and $2$ vote for $a$ at signal $a$ and delegate to one another at signal $b$: $\sigma_1(a)=\sigma_2(a)=a$, $\sigma_1(b)=d_2$, $\sigma_2(b)=d_1$.
\item Nonexperts $3$ and $4$ vote for $b$ at signal $b$ at delegate to the expert at signal $a$: $\sigma_3(a)=\sigma_4(a)=d_e$, $\sigma_3(b)=\sigma_4(b)=b$. 
\item This equilibrium attains EU $0.80272$.
\item One can check that this equilibrium makes the ex-interim correct decision at all signal realizations except $(a,a,a,a,b)$ ;it correctly chooses $b$ at $(b,b,b,b,a)$ (where again, the last signal realization is the expert's and the first four are nonexperts' 1-4 respectively).
\end{itemize}

\subsubsection{Neutral Equilibria Results}

Through this section we will maintain that the committee has no partisans and $\pi=1/2$. Propositions \ref{prop_single_exp} and \ref{prop_multiexp} deal with neutral equilibria.  Before proving them we introduce some definitions and a Lemma. Lemma \ref{lemma_best_neutral} shows that the best neutral strategy profile of LD, DD, or RD in a committee with a symmetric environment is an equilibrium.

\begin{lemma} \label{lemma_best_neutral}
Suppose $\pi=\frac{1}{2}$. For committees without partisans, the best neutral strategy profile under LD, DD, or RD is an equilibrium. 
\end{lemma}

\begin{proof}
We prove this first for LD and then extend it to DD and RD. Fix any voter $i$ and consider $i$'s payoff from action $y$ when her signal is $s_i=a$ and payoff from action $y'$ when her signal is $s_i=b$: 

\begin{align*}
& payoff(s_i=a,y,\sigma_{-i}) =  q_i Pr(\mathcal{O}=A|\omega=A,\sigma_{-i},y) + (1-q_i)Pr(\mathcal{O}=B|\omega=B,\sigma_{-i},y) \\
& payoff(s_i=b,y',\sigma_{-i}) =  (1-q_i) Pr(\mathcal{O}=A|\omega=A,\sigma_{-i},y') + q_i Pr(\mathcal{O}=B|\omega=B,\sigma_{-i},y')
\end{align*}

Note that as $\sigma_{-i}$ prescribes all voters strategies that are symmetric with respect to signal realizations, and hence the state: for any $y \not \in \{a,b\}$ $Pr(\mathcal{O}=A|\omega=A,\sigma_{-i},y)=Pr(\mathcal{O}=B|\omega=B,\sigma_{-i},y)$. Also, $Pr(\mathcal{O}=A|\omega=A,\sigma_{-i},y=a)=Pr(\mathcal{O}=B|\omega=B,\sigma_{-i},y=b)$ and $Pr(\mathcal{O}=A|\omega=A,\sigma_{-i},y=b)=Pr(\mathcal{O}=B|\omega=B,\sigma_{-i},y=a)$. Hence if $y =a $ maximizes $payoff(s_i=a,y,\sigma_{-i})$, then $y'=b$ maximizes $payoff(s_i=b,y',\sigma_{-i})$ and $i$ best responds to $\sigma_{-i}$ by voting sincerely. If $y \in \{d_{j}\}_{j\neq i} \cup \{x\}$ maximizes the payoff at $s_i=a$, the same action does at $s_i=b$. Finally, if $y=b$ maximizes $payoff(s_i=a,y,\sigma_{-i})$, then we must have that $y'=a$ maximizes $payoff(s_i=b,y',\sigma_{-i})$; but this strategy for $i$ is weakly dominated (Lemma 1), and so $i$ must have a best response of either sincere voting, abstaining at both signal realizations, or delegating to the same other voter at both realizations.

The same argument applies to DD with the modification that we don't allow for delegation. For RD, for a representative $i \in \mathcal{J}$, the same argument applies with the modification of not allowing for delegation. For a nonrepresentative $i  \not \in \mathcal{J}$, the same argument applies after replacing  the strategy of voting for $A$ with delegating to partisan representative $a^*$ and replacing voting for $B$ by delegating to $b^*$ (sincere voting is then replicated by delegating to $a^*$ at $s_i=a$ and $b^*$ and $s_i=b$).
\end{proof}

We introduce some more notation needed for the neutral equilibria results. Note that any neutral strategy profile in LD and DD can be described by $(V,X)$ where: (1) $X$ is a set of abstainers $X=\{i \in \mathcal{N}:\sigma_i(a)=\sigma_i(b)=x\}$, and (2) a $V$ is a feasible vote allocation $V=(v_1,...,v_N) \in \mathbb{N}_+^{N}$ with $\sum_i v_i = N$ which is the allocation of votes after delegation (\emph{feasible} because votes are integers and must add to $N$; in DD, the allocation is always $V=(1,...,1)$). Any voter $i \not \in X$ with $v_i \geq 1$ votes sincerely. Given any $(V,X)$, let $\Phi = \{i \in \mathcal{N}: v_i=0\} \cup X$ be the set of nonvoters, and let $U(V,X)$ denote the expected utility delivered by $(V,X)$. Given a strategy profile $(V,X)$ and a subset of voters $\mathcal{J}\subset \mathcal{N}$, let $W^{V,X}(\mathcal{J})=\sum_{i \in \mathcal{J}, i \not \in X} v_i$ be the number of votes cast by $\mathcal{J}$. 


\subsubsection{Proof of Proposition \ref{prop_single_exp}}
Recall Proposition \ref{prop_single_exp}.

\noindent{\bf Proposition \ref{prop_single_exp}}

{\it There is a best neutral equilibrium $\sigma$ of the single expert committee characterized as follows.

\begin{enumerate}
\item If $w^*(e)<2$: all voters vote sincerely under $\sigma$.
\item If $w^*(e) \geq 2$: Exactly $\min\{\floor{w^*(e)}-1,\frac{N-1}{2}\}$ nonexperts delegate to the expert (i.e. $\sigma_i(a)=\sigma_i(b)=d_e$). The remaining voters vote sincerely. 
  
\end{enumerate}

These strategies are independent of the number of nonexperts in the committee. Whenever $w^*(e)>2$, the best neutral equilibrium of LD does strictly better than that of DD.}


In the single expert committee, for any neutral strategy profile $(V,X)$, let $v_e$ be the number of votes the expert holds under $V$. We prove the result via some Lemmas.

\begin{lemma} \label{lemma_no_overweight}
Under LD, the single expert committee has a best neutral strategy profile with: (1) the expert voting sincerely. (2) $1 \leq v_e \leq \min\{ \floor{w^*}, \frac{N+1}{2}\}$. 
\end{lemma}

\begin{proof}
First, note that some voter must vote sincerely under any best neutral strategy profile; if not, all voters delegate or abstain, attaining EU $1/2$; this is dominated by any single voter deviating to vote sincerely. Next, the expert must vote sincerely at any best neutral strategy profile. If not then pick some other voter $i$ who is voting sincerely and is pivotal with positive probability. Consider instead that $e$ and $i$ switched strategies and that all votes delegated to $i$ were delegated to $e$ and vice versa; this new profile does strictly better as whenever $i$ was pivotal, $e$ is now pivotal and is more likely to be correct when voting sincerely. Hence we also have that $v_e \geq 1$. 

Next, any strategy profile with $v_e \geq \frac{N+1}{2}$ and the expert voting sincerely is payoff equivalent to the expert holding exactly $\frac{N+1}{2}$ votes (as the expert dictates the outcome regardless). It remains to show  that $v_e \leq \floor{w^*}$. Consider the strategy profile $(V^*,X^*)$ with $X^*=\emptyset$ (all non-delegators vote sincerely) and $V^*$ s.t. $v_e^* = \floor{w^*}$ and $v_i^* \in \{0,1\}$ for all nonexperts $i$. This strategy profile first-best aggregates the private information of the expert and all $N-\floor{w^*}$ nonexperts who do not delegate, since these voters achieve Nitzan and Paroush weights (while the expert's weight is her rounded down NP weight, this does not change election outcomes relative to first-best; see discussion after Proposition \ref{prop_single_exp}). In any strategy profile $(V',X')$ with $V'$ s.t. $v_e'>\floor{w^*}$, at most $N-\floor{w^*}-1$ nonexperts hold votes and hence this allocation can  best first-best aggregate the information of the expert and $N-\floor{w^*}-1$ nonexperts. Hence $(V^*,X^*)$ must do weakly better.
\end{proof}

\begin{lemma} \label{single_expert_no_ties}
Under LD and DD, in any best neutral equilibrium of the single expert committee or the committee with nonexperts and few experts, there is never a tie in the election.
\end{lemma}

\begin{proof}
There are an odd number of votes total, so if $X = \phi$, there will never be a tie. Now consider a neutral profile $(V,X)$ under which there is a tie with positive probability. Then at least one vote must be cast in abstention and an even number of votes are cast for options $A$ and $B$. Consider instead some voter $i$ in $\Phi$ voting sincerely using a single, previously abstained, vote. Election outcomes will be the same under this new profile except when there was a tie under the old one; conditional on a tie, EU was $1/2$ under the old profile, while under the new one, as $i$'s signal is conditionally independent to others, conditional on the same event EU is $q_i>0.5$.  
\end{proof}

\begin{lemma} \label{single_expert_no_nonexpert_del}
Under LD, the single expert committee has a best neutral strategy profile in which every nonexpert $i \not \in X$ has $v_i \leq 1$.
\end{lemma}

\begin{proof}
Consider any neutral strategy profile $(V,X)$ with $1 \leq v_e \leq \min\{\floor{w^*},\frac{N+1}{2}\}$, $e \not \in X$ (as per Lemma \ref{lemma_no_overweight}), and for contradiction $v_i>1$ for some nonexpert $i \not \in X$. To prove the result, it suffices to find a neutral strategy profile $(V',X')$ with $1 \leq v_e' \leq \min\{\floor{w^*},\frac{N+1}{2}\}$, $e \not \in X$, and $v_j' \leq 1$ for all nonexperts $j$ with $U(V',X') \geq U(V,X)$.

Let $d = |\Phi|$. First suppose that $d \geq \floor{w^*} -1$. Let $V'$ be such that $v_e' = \floor{w^*}$ and $v_j' \leq 1$ for all nonexperts $j$. Let $X'=\phi$, so all vote holders vote sincerely. Then, as in the proof of Lemma \ref{lemma_no_overweight}, $(V',X')$ first-best aggregates the information of $e$ and $N-\floor{w^*}$ nonexperts. $(V,X)$ at best first-best aggregates the information of $N-1-d \leq N-\floor{w^*}$ nonexperts and the expert and hence does weakly worse than $(V',X')$.\footnote{In this case, there is over delegation, resulting in an inefficient loss of information. The profile $(V',X')$ reduces overall delegation, meaning more information is aggregated, and sets $v_e$ to the optimal weight.}

Now suppose $d<\floor{w^*}-1$. Let $M = \{j \in \mathcal{N}: j \not \in \Phi, j \neq e\}$ (set of voting nonexperts). Let $K = W^{V,X}(M)$. $K>|M|$ as some voting nonexpert $i$ has $v_i>1$. Consider neutral strategy profile $(V',X')$ with $X'=X$ and $v'_j =1$ for all $j \in M$, $v'_e = v_e + (K-|M|)$. $V'$ reassigns all `extra' votes of voting nonexperts under $V$ to the expert. Note as $d<\floor{w^*}-1$ that $v_e' < \floor{w^*}$. Let $\Phi' = \{i \in \mathcal{N}: v_i' = 0\} \cup X'$ and note $\Phi' = \Phi$. $(V',X')$ moves the relative weight of the expert closer to the optimal weight, while number of delegators stays the same. It follows that $(V',X')$ does better than $(V,X)$; the remainder of the proof shows this explicitly.

First we introduce some notation: For any neutral $(V'',X'')$ with $\Phi''=\{i \in \mathcal{N}: v_i''=0\} \cup X''$, let $Z(V'',X'') = \{(C,C') \in \mathcal{C}: W^{V'',X''}(C)<W^{V'',X''}(C') \text{ and } W^*(C \cap (\mathcal{N} \setminus \Phi''))>W^*(C' \cap (\mathcal{N} \setminus \Phi''))\}$. $Z$ is the set of signal realizations for which $(V'',X'')$ picks outcome $s_{C'}$ when, given the information of voters not in $\Phi''$ (i.e. active voters), $s_C$ would have been better ex-interim.

We will show $Z(V',X') \subset Z(V,X)$. As $\Phi=\Phi'$ and all signal realizations occur with positive probability, this is sufficient to show $(V',X')$ is weakly better than $(V,X)$. Under $(V',X')$, the nonexperts outside $\Phi'$ have optimal weights of $1$ while the expert is underweighted ($v_e < \floor{w^*}$). Hence $Z(X',V')$ comprises of realizations at which the expert's coalition looses when it `should' win.

\begin{align*}
 & Z(V',X') = \{(C,C') \in \mathcal{C}: e \in C, \text{ } W^*(C \cap (\mathcal{N} \setminus \Phi')) > W^*(C' \cap (\mathcal{N} \setminus \Phi')) \\
 & = W^{V',X'}(C') > W^{V,X}(C) \}
\end{align*}

where the first inequality says that $C$ has better information and the second says that $C'$ wins the election (note ties are ruled out by Lemma \ref{single_expert_no_ties}.) Now consider any $(C,C') \in Z(V',X')$. Note that $W^{V,X}(C')\geq W^{V',X'}(C')$ as $X'=X$ and $C'$ contains only nonexperts and all nonexperts possess weakly less votes under $V'$ than $V$. Also, $W^{V',X'}(C) \geq W^{V,X}(C)$ as moving from $V$ to $V'$ entails more delegation to the expert, who is in $C$. This implies that $(C,C') \in Z(V',X') \implies (C,C') \in Z(V,X)$. 

\end{proof}

Now we prove Proposition \ref{prop_single_exp}.

\begin{proof}
Consider any neutral strategy profile $(V,X)$ satisfying the properties of Lemmas \ref{lemma_no_overweight}, \ref{single_expert_no_ties}, \ref{single_expert_no_nonexpert_del}. That is:  $1 \leq v_e \leq \min\{\floor{w^*},\frac{N+1}{2}\}$, $e \not \in X$, and $v_i \leq 1$ for all nonexperts $i \not \in X$. Let $\Phi = \{i \in \mathcal{N}: v_i=0\} \cup X$ and $M = \{i \in \mathcal{N}: i \not \in \Phi, i \neq e\}$ (set of voting nonexperts, each with one vote). If $v_e > |M|$ then the expert dictates the election outcome; the strategy profile described in the Proposition must do weakly better as it first-best aggregates the information of the expert and $\min\{\floor{w^*}-1,\frac{N+1}{2}-1\}$ nonexperts. 

Assume henceforth $v_e \leq M$. By Lemma \ref{single_expert_no_ties}, $|M|+v_e$ must be odd. Fix some $i \in M$ and consider the incentive of $i$ to delegate her vote to the expert (at both possible signal realizations $s_i$). Delegation only changes election outcomes when: $i$'s vote is pivotal. Let $piv_i$ be the event $i$'s vote is pivotal; this occurs at signal realizations at which: (1) $\frac{|M|+v_e -1}{2}$ other nonexperts receive signals disagreeing with $e$, and (2) $\frac{|M|+v_e -1}{2} - v_e$ other nonexperts receive signals agreeing with $e$ (so, excluding $i$, there is a tie). Note that $i$ being pivotal is a function of the signal realizations of voters; as all voting is sincere, the prior is flat, and voters receive conditionally independent signals, if $i$ keeps their vote, the expected payoff conditional on $piv_i$ is $q$. If the vote is given to the expert, then the expected payoff conditioning on pivotality is the conditional probability the expert receives the correct signal (as they vote sincerely). Delegation to the expert weakly increases EU whenever:

\begin{align*}
& Pr(s_e=\omega |piv_i) \geq q \iff \frac{Pr(piv_i|s_e=\omega)Pr(s_e=\omega)}{Pr(piv_i|s_e=\omega)Pr(s_e=\omega)+Pr(piv_i|s_e\neq\omega)Pr(s_e\neq\omega)} \geq q \\
& \iff \frac{pq^{\frac{|M|+v_e -1}{2}-v_e}(1-q)^{\frac{|M|+v_e -1}{2}}}{pq^{\frac{|M|+v_e -1}{2} - v_e}(1-q)^{\frac{|M|+v_e -1}{2}} + (1-p)(1-q)^{\frac{|M|+v_e -1}{2}-v_e}q^{\frac{|M|+v_e -1}{2}}} \geq q \\
& \iff p \geq \frac{q^{v_e+1}}{q^{v_e+1}+(1-q)^{v_e+1}} \iff v_e+1 \leq w^* 
\end{align*}

Where the second line is due to  all voters voting sincerely, implying: $Pr(piv_i|s_e=\omega) = \binom{|M|-1}{\frac{|M|+v_e -1}{2}}q^{\frac{|M|+v_e -1}{2} - v_e}(1-q)^{\frac{|M|+v_e -1}{2}}$ and $Pr(piv_i|s_e \neq \omega) = \binom{|M|-1}{\frac{|M|+v_e -1}{2}-v_e}q^{\frac{|M|+v_e -1}{2}}(1-q)^{\frac{|M|+v_e -1}{2}-v_e}$.  The last equivalence is derived by rearranging and taking logs. The expression implies that nonexperts in $M$ will have an incentive to delegate to the expert unless we have: $v_e = \floor{w^*}$. Such a strategy profile will first-best aggregate the information of the expert and $|M| \leq N-\floor{w^*}$ nonexperts, which does weakly worse than the strategy profile described in the Proposition.
\end{proof}

\subsubsection{Proof of Proposition \ref{prop_multiexp}}
Consider a committee in which all members are independents and the prior is $\pi = \frac{1}{2}$. Suppose there are $n$ (even) voters who each have precision $q_i=q$; call these voters the `nonexperts'. The remaining $E=N-n$ (odd) voters are `experts'; each such voter $i$ has precision $q_i>q$ and $w^*(i) = \frac{\log(\frac{q_i}{1-q_i})}{\log(\frac{q}{1-q})}$.\footnote{Assuming $n$ even and $E$ odd is for convenience.}  Recall Proposition \ref{prop_multiexp}: 

\noindent {\bf Proposition \ref{prop_multiexp}}
{\it In the committee described above, for large enough $n$ there is a best neutral equilibrium of LD $\sigma$ in which: \begin{enumerate}
\item For each expert $i$, exactly $\floor{w^*(i)}-1$ nonexperts delegate to her (i.e. for  $\floor{w^*(i)}-1$ nonexperts $j$, $\sigma_j(a)=\sigma_j(b)=d_i$). 
\item All nondelegators vote sincerely.

\end{enumerate}

Moreover the best neutral equilibrium in LD is strictly better than the best neutral equilibrium in DD whenever $w^*(i)>2$ for some expert $i$.}

For the following results, we will consider fixing a set of $E$ experts $\mathcal{E}$ and growing the number of nonexperts $n$. For $n = 2,4,6,...$, let $(V_n,X_n)$ be the best neutral equilibrium for the committee with $n$ nonexperts. At this equilibrium, let $\Phi_n = \{i \in \mathcal{N}: v_{n,i}=0\} \cup X_n$ be the set of nonvoters, $S_n = \mathcal{N} \setminus \Phi_n$ be the set of sincere voters, and $S^E_n$ and $S^n_n$ be the set of sincere voting experts and nonexperts respectively. Note that by an similar argument to point (1) Lemma \ref{lemma_no_overweight}. $S^E_n$ is nonempty.

The proof of Proposition 4 proceeds as follows. Lemma \ref{lemma_wasting} next shows that the number of delegators and abstainers under the best neutral strategy profile remains bounded as $n$ gets large. The is because the amount of information `wasted' by delegation/abstention cannot be unbounded. An implication of this is that as $n$ increases, under the sequence of best neutral equilibria, the share of all votes held by nonexperts holding only a single vote converges to $1$. Lemma \ref{lemma_multiexp} then shows that under these circumstances nonexperts have incentives to delegate to experts until experts attain their rounded down Nitzan and Paroush optimal weights, as this ensures all experts are optimally weighted relative to single vote holding nonexperts, who dominate the committee. Proposition 4 then follows from this.

\begin{lemma} \label{lemma_wasting}
There exists an $M$ s.t. $|\Phi_n|<M$ for all $n$. 
\end{lemma}

We first prove a helpful claim. For any $k \in \mathbb{N}$, let $\Pi(k)$ be the experiment induced by observing $k$ nonexpert signal realizations.  Let $\Pi^M(k)$ be the garbling of $\Pi(k)$ which just reports  whether a majority of the nonexpert signals realized to $a$ or $b$. Let $\Pi(\mathcal{J})$ be the experiment induced by observing signal realizations of $\mathcal{J} \subset \mathcal{N}$ and for any expert $e$, let $\Pi(e)$ denote the experiment induced by observing just signal realization $s_e$.

\emph{Claim.} For large enough $n$, $\Pi(n)$ is Blackwell more informative (Blackwell, 1953) than $\Pi(\mathcal{E})$.

\begin{proof}

We prove that for any single expert $e \in \mathcal{E}$, there exists a number of nonexperts $n_e$ such that $\Pi(n_e)$ is Blackwell more informative than $\Pi(e)$; it follows that for some finite $n$, $\Pi(n)$ is Blackwell more informative than $\Pi(\mathcal{E})$. 

The posterior upon observing the realization $s^n$ of $\Pi^M(n)$ is: $Pr(\omega=A|s^n=A) = Pr(s^n=A|\omega=A)Pr(\omega=A)/Pr(s^n=A)=Pr(s^n=A|\omega=A)$ (where the last equality is due to the symmetric prior and signal structures). By the Condorcet Jury Theorem, $Pr(s^n=A|\omega=A) \rightarrow 1$ as $n \rightarrow \infty$ ( similarly, $Pr(s^n=B|\omega=B) \rightarrow 1$). Hence, for large enough $n_e$, the distribution of posteriors under $\Pi^M(n_e)$ is a mean-preserving spread of the distribution of posteriors induced by $\Pi(e)$ (under $\Pi(e)$, the posteriors realizes to $q_e$ and $1-q_e$ with equal probability). Hence for large enough $n_e$, $\Pi(n_e)$ is Blackwell more informative than $\Pi^M(n_e)$ which is Blackwell more informative than $\Pi(e)$.
\end{proof}

Now we prove Lemma \ref{lemma_wasting}.

\begin{proof}
Suppose not for contradiction. Then there exists a subsequence of $n$, $\{n_k\} \rightarrow \infty$ s.t. $|\Phi_{n_k}| \rightarrow \infty$. Note that this implies the number of nonvoting nonexperts is goes to infinity (as $E$ finite).

Note that at best, $(V_n,X_n)$ first-best aggregates the information of $S_n$; the information of voters in $\Phi_n$ is not taken into account. If $V_n$ did first-best aggregate the information of $S_n$, this would do exactly as well as a decision maker with the same preferences as all voters who observes exactly $S^n_n$ nonexpert signals and the signals of $S^E_n$ before making a decision.

As $|\Phi_{n_k}| \rightarrow \infty$, by the Claim above, for large enough $k$, $\Pi(|\Phi_{n_k}|)$ is strictly Blackwell more informative than $\Pi(\mathcal{E})$. But then we have that $\Pi(|\Phi_{n_k} \cup S^n_{n_k}|)$ is strictly more informative than $\Pi(\mathcal{E} \cup S^n_{n_k})$. Under $(V_n,X_n)$, voters do weakly worse than if a benevolent DM made decision after observing information from $\Pi(S^n_{n_k})$ and $\Pi(\mathcal{E})$. But the DM would do strictly better if she observed $\Pi(|\Phi_{n_k} \cup S^n_{n_k}|)$ instead.

Consider instead the neutral strategy profile under which all $N$ voters vote sincerely under the initial allocation. As $E>1$, this does strictly better than $N$ nonexperts voting sincerely, which does weakly better than a committee of $|\Phi_{n_k} \cup S^n_{n_k}|$ nonexperts voting sincerely, which does as well as the benevolent DM observing $\Pi(|\Phi_{n_k} \cup S^n_{n_k}|)$ before making a decision (Nitzan and Paroush, 1982, Corollary 2), which in turn is weakly better than our supposed efficient equilibrium. Contradiction.

\end{proof}

The Lemma states that (under efficiency) the number of sincerely voting voters goes to infinity as $n$ does. It also implies (under efficiency) that there is a number of nonexperts, $Y_n \rightarrow \infty$, that each hold a single vote and vote sincerely.

Let $\{(V_n,X_n)\}_{n=2,4,...}$ be a sequence of neutral strategy profiles of the game with $n$ nonexperts and experts $\mathcal{E}$. Let $\Phi_n = \{i: v_{n,i}=0\} \cup X_n$ be the set of nonvoters, $\mathcal{Y}_n = \{i: q_i=q, v_{n_i}=1, i \not \in X_n\}$ be the set of voting nonexperts holding one vote, and $Y_n = |\mathcal{Y}_n|$. Finally let $T_n = (E+n) - \sum_{i \in \Phi_n} v_i$ be the number cast for the two alternatives (i.e. not abstained on) under $(V_n,X_n)$. Proposition \ref{prop_multiexp} will follow from the following Lemma.

\begin{lemma} \label{lemma_multiexp}
Suppose $T_n \rightarrow \infty$, $Y_n / T_n \rightarrow 1$, and for all $n$, ties occur w.p. $0$ and all experts vote sincerely under $(V_n,X_n)$. There exists $\bar{n}$ s.t. $\forall n \geq \bar{n}$ s.t.: (a) if $v_{n,e}<\floor{\frac{\log(q_e/1-q_e)}{\log(q/1-q)}}$ for some $e \in \mathcal{E}$ then it is weakly profitable for some $i \in \mathcal{Y}_n$ to addtionally delegate to $e$ at both signal realizations. (b) if $v_{n,e}>\floor{\frac{\log(q_e/1-_e)}{\log(q/1-q)}}$ for some $e \in \mathcal{E}$ then there is a weakly better neutral strategy profile with one vote from $e$ reassigned to some nonexpert $i$ with $v_{n,i}=0$.
\end{lemma}

\begin{proof}
First we prove (a). Take an infinite increasing sequence of $n$ and strategy profiles $(V_n,X_n)$ satisfying the three conditions and suppose: (1) there is some expert $e$ s.t. for each $n$ along the sequence $v_{n,e}<\floor{\frac{\log(q_e/1-q_e)}{\log(q/1-q)}}$ (we can assume the expert in violation is the same along the sequence as $E$ is finite and so if not, we can find an appropriate subsequence), (2) $Y_n/n \rightarrow 1$. Assume $\mathcal{Y}_n$ is nonempty $\forall n$ along the sequence (or else we can take an appropriate subsequence). As all nonexperts are identical, we fix some nonexpert $i$ and assume $i \in \mathcal{Y}_n$ $\forall n$.

\emph{Some definitions.} Let $\mathcal{C}_n$ be the set of possible signal realizations for the committee with $n$ nonexperts. Given $(V_n,X_n)$, we say signal realizations $(C_n,C_n')$ and $(D_n,D_n')$ are equivalent to $i$ in the committee with $n$ nonexperts if: $i \in C_n \iff i \in D_n$, for all $j \not \in \mathcal{Y}_n$ $j \in C_n \iff j \in D_n$, $|\{j \in \mathcal{Y}_n : j \in C_n\}| =|\{j \in \mathcal{Y}_n : j \in D_n\}|$, and $|\{j \in \mathcal{Y}_n : j \in C_n'\}| =|\{j \in \mathcal{Y}_n : j \in D_n'\}|$. In other words, $(C_n,C_n')$ and $(D_n,D_n')$ are identical signal realizations up to permuting the labels of voting nonexperts holding $1$ vote other than $i$. Events $(C_n,C_n')$ and $(D_n,D_n')$ are equally probable. The expected utility at $V$ from $i$ keeping her own vote and casting it sincerely is the same conditional on realization $(C_n,C_n')$ as it is under $(D_n,D_n')$; this is because $C_n$ and $D_n$ contain exactly the same informational content, as do $C_n'$ and $D_n'$. Similarly, $i$'s expected utility from delegating to any $e' \in \mathcal{E}$ is the same conditional on $(C_n,C_n')$ as on $(D_n,D_n')$. For any $n$, $\mathcal{C}_n$ can be partitioned into `classes' of signal realizations $\mathcal{C}_{n,1},...,\mathcal{C}_{K,n}$, where $(C_n,C_n'),(D_n,D_n') \in \mathcal{C}_{n,k}$ if and only if $(C_n,C_n')$ and $(D_n,D_n')$ are equivalent to $i$.

For any $n$ on the sequence consider any class $\mathcal{C}_{n,k}$ of realizations equivalent to $i$ under $(V_n,X_n)$. Take any $(D_n,D_n') \in \mathcal{C}_{n,k}$ with $i \in D_n$, $e \in D_n'$. Let $E_{D}^n = \{e' \in \mathcal{E}: e' \in D_n\}$ be all the experts in $D_n$ and $E_{D'}^n = \mathcal{E} \setminus (E_D^n \cup \{e\})$ be the same for $D_n'$ excluding $e$. Let $M_D^n$ ($M_{D'}^n$) be the number of nonexperts in $\mathcal{Y}_n$ that are in $D_n$ ($D_n'$), excluding $i$. By definition, $E_{D}^n,E_{D'}^n,M_D^n,M_{D'}^n$ are the same for all $(D_n,D_n') \in \mathcal{C}_{n,k}$ and are different (in at least one component) for all $(C_n,C_n') \not \in \mathcal{C}_{n,k}$; these four sets uniquely define $\mathcal{C}_{n,k}$. 

Now we return to the proof. For any $n$ along the sequence and consider any signal profile $(D_n,D_n')$ for which the delegation from $i$ to $e$ is bad under $(V_n,X_n)$ \textemdash it strictly decreases EU relative to $i$ voting sincerely. At any such event, we have: $i \in D_n$, $e \in D_n'$ ($i$ and $e$ disagree), $Pr(s_{D_n} = \omega)>1/2$ ($s_{D_n}$ is more likely to be correct and so $i$ is better off casting her own vote), $W^{V_n,X_n}(D_n) = \frac{T_n+1}{2} = W^{V_n,X_n}(D_n') + 1$ (as there are no ties under $(V_n,X_n)$, $T_n$ is odd and this is necessary for $i$'s vote to be pivotal). Let $\mathcal{C}_{n,k} \ni (D_n,D_n')$ be the class of $(D_n,D_n')$ and note that delegation from $i$ to $e$ is equally bad for all signal realizations in this class. 

For each such $\mathcal{C}_{n,k}$, we will show that for large enough $n$, there is a corresponding \emph{unique} class of signal realizations, $\mathcal{C}_{n,l}$, for which delegation is good. That is, for all $(F_n,F_n') \in \mathcal{C}_{n,l}$, $i \in F_n, e \in F_n'$ and  $W^{V_n,X_n}(F_n) = \frac{T_n+1}{2} = W^{V_n,X_n}(F_n')+1$ ($i$'s vote is pivotal) and $Pr(s_{F_n}' = \omega)> \frac{1}{2}$. Further, we will show that $Pr(s_{F_n'} = \omega) > Pr(s_{D_n} = \omega)$ (implying that the EU gain from delegation conditional on $\mathcal{C}_{n,l}$ exceeds loss under $\mathcal{C}_{n,k}$) and $Pr(\mathcal{C}_{n,l})>Pr(\mathcal{C}_{n,k})$ (i.e. a class of signal realizations in which delegation is good occurs with larger probability). Together, this implies that delegation is profitable. The remainder of the proof constructs $\mathcal{C}_{n,l}$ satisfying the above properties.

Fix a $\mathcal{C}_{n,k}$ for which delegation is bad and consider any $(D_n,D_n') \in \mathcal{C}_{n,k}$. Note that or large enough $n$, $M_{D}^n = (Y_n-1) - M_{D'}^n \geq (Y_n-1) - (\frac{T_n}{2} - 1 - v_{n,e})>v_{n,e}$; the first inequality holds because there may be experts other than $e$ in $D'$, and the second holds because $Y_n \geq \frac{T_n}{2}+1$ for large enough $n$ (as $Y_n/T_n \rightarrow 1$). Consider any such large enough $n$.

Define $\mathcal{C}_{n,l}$ as follows. Let each $(F_n,F_N') \in \mathcal{C}_{n,l}$ be such that $F_n$ contains $E_{D'}^n$, $M_{D'}^n+v_{n,e}$ nonexperts (excluding $i$), and $i$. $F_n'$ contains $E_{D}^n$, $M_{D}^n - v_{n,e}$ nonexperts, and $e$. Then we have $W^{V_n,X_n}(F) = \frac{T_n+1}{2}=W^{V_n,X_n}(F')+1$. 

Let $r_{E_{D}^n} = \Pi_{e' \in E_{D}^n} r_{e'}$, $w_{E_{D}^n} = \Pi_{e' \in E_{D}^n} (1-r_{e'})$ and similar for $r_{E_{D'}^n},w_{E_{D'}^n}$. Note that:

\begin{align*}
& \frac{Pr(s_{F_n'} = \omega)}{Pr(s_{F_n'} \neq \omega)} \frac{Pr(s_{D_n} \neq  \omega)}{Pr(s_{D_n} = \omega)}  = \frac{r_e r_{E_{D}^n} q^{M_{D}^n-v_{n,e}} (1-q) w_{E_{D'}^n} (1-q)^{M_{D'}^n+v_{n,e}}}{(1-r_e) w_{E_{D}^n} (1-q)^{M_{D}^n-v_{n,e}} q r_{E_{D'}^n} q^{M_{D'}^n+v_{n,e}}} \frac{(1-q) w_{E_{D}^n} (1-q)^{M_{D}^n} r_e r_{E_{D'}^n} q^{M_{D'}^n}}{q r_{E_{D}^n} q^{M_{D}^n} (1-r_e) w_{E_{D'}^n} (1-q)^{M_{D'}^n} } \\
& = \frac{r_e^2 q^{M_{D}^n + M_{D'}^n - v_{n,e}} (1-q)^{M_{D}^n + M_{D'}^n+v_{n,e}+2}}{ (1-r_e)^2 q^{M_{D}^n + M_{D'}^n+v_{n,e}+2} (1-q)^{M_{D}^n + M_{D'}^n - v_{n,e}}} \\
& = \left[ \frac{r_e (1-q)^{v_{n,e}+1}}{(1 - r_e) q^{v_{n,e}+1}} \right]^2 
\end{align*}

This expression is $> 1$ if and only if:

\begin{align*}
& \log(\frac{r_e}{1-r_e}) > (v_{n,e}+1) \log(\frac{q}{1-q}) \iff w^* > v_{n,e}+1
\end{align*}

This holds and so $\frac{Pr(s_{F_n'} = \omega)}{Pr(s_{F_N'} \neq \omega)} \frac{Pr(s_{D_n} \neq  \omega)}{Pr(s_{D_n} = \omega)} >1$  which implies $Pr(s_{F_n'} = \omega) > Pr(s_{D_n} = \omega) $. It remains for us to show that for large $n$, $\mathcal{C}_{n,l}$ is more probable than $\mathcal{C}_{n,k}$:

\begin{align*}
& \frac{Pr(\mathcal{C}_{n,k})}{Pr(\mathcal{C}_{n,l})} \\
& = \frac{\binom{Y_n-1}{M_{D}^n}[r_{E_{D}^n} q^{M_{D}^n + 1} w_{E_{D'}^n}(1-r_e) (1-q)^{Y_n - 1 - M_{D}^n} + w_{E_{D}^n} (1-q)^{M_{D}^n + 1} r_{E_{D'}^n} r_e  q^{Y_n - 1 - M_{D}^n} ]}{\binom{Y_n-1}{ M_{D}^n-v_{n,e}} [r_{E_{D'}^n} q^{Y_n - 1 - M_{D}^n + v_{n,e} + 1} w_{E_{D}^n} (1-r_e) (1-q)^{M_{D}^n-v_{n,e}} + w_{E_{D'}^n} (1-q)^{Y_n - 1 - M_{D}^n + v_{n,e} + 1} r_{E_{D}^n} r_e q^{M_{D}^n-v_{n,e}}]} \\
& = \frac{\binom{Y_n-1}{M_{D}^n}}{\binom{Y_n-1}{ M_{D}^n-v_{n,e}}} \\
& \times \frac{[r_{E_{D}^n} q^{M_{D}^n -v_{n,e}} w_{E_{D'}^n}(1-q)^{Y_n - 1 - M_{D}^n} (q^{v_{n,e}+1} (1-r_e))  + w_{E_{D}^n} (1-q)^{M_{D}^n -v_{n,e}} r_{E_{D'}^n}   q^{Y_n - 1 - M_{D}^n} ((1-q)^{v_{n,e}+1} r_e)]}{ [w_{E_{D}^n} (1-q)^{M_{D}^n -v_{n,e}} r_{E_{D'}^n}   q^{Y_n - 1 - M_{D}^n} (q^{v_{n,e}+1} (1-r_e))  + r_{E_{D}^n} q^{M_{D}^n -v_{n,e}} w_{E_{D'}^n}(1-q)^{Y_n - 1 - M_{D}^n} ((1-q)^{v_{n,e}+1} r_e)]} \\
\end{align*}

Let $a = r_{E_{D}^n} q^{M_{D}^n -v_{n,e}} w_{E_{D'}^n}(1-q)^{Y_n - 1 - M_{D}^n}$, $b=w_{E_{D}^n} (1-q)^{M_{D}^n -v_{n,e}} r_{E_{D'}^n}   q^{Y_n - 1 - M_{D}^n}$, $c=(1-q)^{v_{n,e}+1}r_e$ and $d = q^{v_{n,e}+1} (1-r_e)$. Note that $c > d$ as $e$'s signal is strictly more informative than $v_{n,e}+1$ nonexperts and so $e$ is more likely to receive the correct signal.\footnote{More formally: $c >d \iff \log(c)>\log(d) \iff log(\frac{r_e}{1-r_e}) > (v_{n,e} + 1) \log (\frac{q}{1-q})$ which is true.} Similarly $a>b$ as $D$ is collectively more informative than $D'$ and so $D \setminus \{v_{n,e} + 1 \text{ nonexperts}\}$ must be more informative than $D' \setminus \{e\}$.\footnote{Formally: $a>b \iff \log(a) > \log(b) \iff \sum_{k \in E_{D}^n} \log(r_k/1-r_k) + (M_{D}^n - v_{n,e})\log(q/1-q) > \sum_{k \in E_{D'}^n} \log(r_k/1-r_k) + (Y_n-1-M_{D}^n)\log(q/1-q)$. This is true as we have $LHS + (v_{n,e} + 1) \log(q/1-q) > RHS + \log(r_e/1-r_e)$ but $(v_{n,e} + 1) \log(q/1-q) < \log(r_e/1-r_e)$} Hence the ratio of the bracketted terms in the expression above is: 

\begin{align*}
& \frac{ad + bc}{bd + ac} < 1
\end{align*}

Note that $M^n_D = \frac{T_n+1}{2}-1-W^{X_n,V_n}(E^n_d)$ and $M^n_{D'}=Y_n-1-M^n_D = \frac{T_n-1}{2} - W^{X_n,V_n}(E^n_{D'}) - v_{n,e}$; rearranging: $Y_n = 2 M_{D}^n + 1 - v_{n,e} + W^{V_n,X_n}(E_{D}^n) - W^{V_n,X_n}(E_{D'}^n)$. Substituting this for $Y_n$ in $a$ and $b$ and doing some algebra, we can see that the ratio above is independent of $Y_n$ and $M_D^n$; hence the only dependence on $n$ is through $r_{E^n_D},w_{E^n_D},r_{E^n_{D'}},w_{E^n_{D'}}$, and $v_{n,e}$. As $\mathcal{E}$ contains finitely many experts (who can hence be split into $D_n$ and $D_n'$ in finitely many ways), and $0<v_{n,e} < \floor{\frac{\log(q_e/1-q_e)}{\log(q/1-q)}}$, the LHS of the ratio can only finitely many values and hence is bounded away from $1$. We next show that $\frac{\binom{Y_n-1}{M_{D}^n}}{\binom{Y_n-1}{ M_{D}^n-v_{n,e}}}  \rightarrow 1$, implying that $\frac{Pr(\mathcal{C}_{n,k})}{Pr(\mathcal{C}_{n,l})}<1$ for all $n$ large enough.

Note that $Y_n = 2 M_{D}^n + 1 - v_{n,e} + W^{V_n,X_n}(E_{D}^n) - W^{V_n,X_n}(E_{D'}^n)$ also implies that $Y_n /M_{D}^n \rightarrow 2$ as $n \rightarrow \infty$ (divide the expression by $Y_n$; as $Y_n/T_n \rightarrow 1$, $\frac{1-v_{n,e} + W^{V_n,X_n}(E_{D}^n) - W^{V_n,X_n}(E_{D'}^n) }{Y_n} \rightarrow 0$). Hence: 

\begin{align*}
& \frac{\binom{Y_n-1}{M_{D}^n}}{\binom{Y_n-1}{M_{D}^n-v_{n,e}}} = \frac{(Y_n - 1 -M_{D}^n +v_{n,e})! (M_{D}^n-v_{n,e})!}{M_{D}^n! (Y_n-1-M_{D}^n)!} \\
& = \frac{(1/M_{D}^n)^{v_{n,e}}\Pi_{k=1}^{v_{n,e}} (Y_n-1-M_{D}^n+k)}{(1/M_{D}^n)^v_{n,e}\Pi_{k=1}^{v_{n,e}} (M_{D}^n - v_{n,e} +k)} = \frac{\Pi_{k=1}^{v_{n,e}} (Y_n/M_{D}^n-1/M_{D}^n-1+k/M_{D}^n)}{\Pi_{k=1}^{v_{n,e}} (1 - v_{n,e}/M_{D}^n +k/M_{D}^n)} \rightarrow 1
\end{align*}

This concludes the proof of Case (a).

Case (b) follows from the same working. Take such a sequence of $n$ and neutral strategy profiles $(V_n,X_n)$ satisfying the conditions in the Lemma and s.t. there exists some expert $e \in \mathcal{E}$ with efficient $v_{n,e}>\floor{\frac{\log(r_e/1-r_e)}{\log(q/1-q)}}$ for all $n$ along the sequence. Consider some nonexpert $i$ with $v_{n,i}=0$ $\forall n$ and consider the sequence of neutral strategy profiles $(V_n',X_n')$ with $X_n'=X_n$ $\forall n$, and $\forall n$: $v_{n,i}'=1$, $v_{n,e}' = v_{n,e}-1$ and $v_{n,j}'=v_{n,j}$ for all $j \neq i,e$. Note that in the proof for case (a), delegation from $i$ to $e$ was strictly profitable (for large enough $n$) \emph{if and only if} $v_{n,e}' < \floor{\frac{\log(r_e/1-r_e)}{\log(q/1-q)}}$. As $v_{n,e}' = v_{n,e} - 1 \geq \floor{\frac{\log(r_e/1-r_e)}{\log(q/1-q)}}$, it strictly decreases EU for $i$ to delegate to $e$ under $(V_{n}',X_n')$ for large enough $n$ or, in other words, it is strictly profitable for $e$ to delegate a single vote to $i$ under $V_n$ (for large enough $n$).

\end{proof}


%

\textbf{Proof of Proposition 4.} 

\begin{proof}
First note that Lemma \ref{single_expert_no_ties} applies in this setting as well, by the same argument. Hence ties occur w.p. $0$ under the best neutral equilibrium. By, Lemma \ref{lemma_wasting} $Y_n \rightarrow \infty \implies T_n \rightarrow \infty$ and $Y_n / T_n \rightarrow 1$ under any sequence of best neutral equilibria $(V_n,X_n)$. Further, for large enough $n$, all experts vote sincerely under $(V_n,X_n)$; this is because for $n$ large enough $Y_n > 0$, and hence if some expert $e \in \mathcal{E}$ were abstaining, payoffs would be improved by swapping the vote allocation and voting strategies of $e$ and some $i \in \mathcal{Y}_n$. Lemma \ref{lemma_multiexp} then implies part (1) of the Proposition. Note that the argument of Lemma \ref{lemma_multiexp} also directly implies that any sincere voting nonexpert $j$ must have $v_{n,j}=1$ for large enough $n$; to see this, if $v_{n,j}>1$, then using part (b) of the Lemma with $j$ instead of $b$ implies that payoffs could be improved by reassigning a vote from $j$ to some nonexpert $i$ with $v_{n,i}=0$, and having $i$ vote sincerely (nothing about the proof required that $e$ was an expert). 

Finally, to complete the proof of part (2) we need to show that for large enough $n$, all nondelegating nonexperts vote sincerely, or equivalently that no nonexperts abstain. To see this, first suppose that for some $n$, the best neutral equilibrium $(V_n,X_n)$ has $X_n \phi$ and satisfies all properties in the previous paragraphs. Now consider adding two nonexperts $i,j$ to the committee and suppose all other voters are adopting the same strategies as under $(V_n,X_n)$. Compare the welfare from $i,j$ both abstaining versus both voting sincerely. The two strategies only differ when $s_i = s_j$ (if $s_i=s_j$, then in either case no net extra votes are cast for either alternative). Conditional on $s_i=s_j$, collectively $i$ and $j$ have private information equivalent to that of a single voter $k$ with a precision such that her optimal weight is $w^*(k)=2$. $i$ and $j$ also collectively hold two votes. Lemma \ref{lemma_multiexp} tells us that when $n$ is large enough, conditional on $s_i=s_j$, if voters $i$ and $j$ were replaced with voter $k$, voter $k$ should optimally vote sincerely with two votes (and this is better than $k$ abstaining). Hence, it is better for $i$ and $j$ to both vote sincerely than to abstain. 

Finally, to show that no abstention is optimal for large enough $n$, suppose for some large $n$ that the best neutral equilibrium $(V_n,X_n)$ has $X_n \neq \phi$. Then it must be that $X_n$ contains only nonexperts (all experts vote sincerely at best equilibrium), $X_n$ wlog contains only nonexperts each with $v_{n,i}=1$ (if $v_{n,i}>1$ for some abstaining nonexpert, we can find a payoff equivalent strategy profile by distributing the additional votes to other nonexperts who abstain), and $X_n$ contains an even number of voters (Lemma \ref{single_expert_no_ties}). Let $m = |X_n|$. Then for committee with $n'=n-m$, the best neutral equilibrium $(V_{n'},X_{n'}$) must have $X_n'=\phi$ and must be otherwise identical to $(V_n,X_n)$ (because: (1) we can replicate the outcomes of $(V_n,X_n)$ as the additional $m$ abstaining nonexperts played no role, and (2) the best equilibrium can do no better with less voters). Similarly, for the committee with $n'+2$ nonexperts, the best neutral equilibrium must have two abstaining nonexperts and be outcome equivalent to $(V_n,X_n)$. But for large enough $n$, if $(V_{n'},X_{n'})$ has $X_{n'}=\phi$, the best equilibrium with $n'+2$ nonexperts must also have $X_{n'+2}=\phi$ (see previous paragraph). Contradiction.

That the best neutral equilbrium of LD does strictly better than that of DD for large $n$ when some expert $e$ has $\floor{w^*(e)}>1$: the best neutral equilibrium of DD is sincere voting for all voters. By Lemma \ref{lemma_multiexp}, under this strategy profile this is a strict incentive for some nonexpert to delegate to $e$.
\end{proof}

 \subsection{Proof of Section 5.}
 
 \noindent {\bf Proof of Proposition \ref{prop: dominant}
 Below, we denote the number of independents who have type $(I,q_i)$, $q_i<1$ as $n_U$.}
 
 \begin{prop}
 Suppose that $n_e \geq 1$, $n_U \geq n_e+ |n_A-n_B|+1$ $n_I \geq |n_A-n_B|$, and $max(n_e+n_A,n_e+n_B) \leq \frac{N}{2}$.
Then  the election game with delegation is dominance solvable. The DS outcome is the efficient outcome. 
\end{prop}

\begin{proof}
Consider partisan voters: it is obvious that $\sigma_i(a)=\sigma_i(b)=a$ ($A$ voters)  or  $\sigma_i(a)=\sigma_i(b)=b$, ($B$ voters) is weakly better than all other strategies, as it maximizes the probability of the preferred outcome for any strategy profile $\sigma_{-i}.$   
W.l.o.g consider a type $A$ voter. Suppose $N$ is odd, then  it is strictly better to vote unresponsively for $A$ than  voting unresponsively for $B$ or abstaining unresponsively for a strategy profile  where $N-1$ voters are equally divided between $A$ and $B$.  Consider delegating to any other voter $j\neq i$: there exists a profile for each such $j$ such that $\sigma_j(a)=\sigma_j(b)= b$. This is strictly worse than $\sigma_i(a)=\sigma_i(b)=a$.
If $N$ is even, then let $\frac{N-2}{2}$ vote unresponsively $a$ and $\frac{N-2}{2}+1$ vote $b$ unresponsively - then it is a strict best response to vote unresponsively  $a$, as it creates a tie. $\sigma_i(a)=\sigma_i(b)=b$ or $\sigma_i(a)=\sigma_i(b)=x$  will lead to $B$ for sure.  Consider delegating to any other voter $j\neq i$: there exists a profile for each such $j$ such that $\sigma_j(a)=\sigma_j(b)= b$. This is strictly worse than $\sigma_i(a)=\sigma_i(b)=a$.

Consider an independent (fully informed) expert $i\in N_e$: it is  obvious that the sincere strategy is weakly better than any other as it maximizes the probability of getting the correct alternative for any $\sigma_{-i}.$   Suppose $N$ is odd, then  it is strictly better to vote sincerely for a strategy profile where $N-1$ voters are equally divided between $a$ and $b$ relative to voting unresponsively, or semi- sincerely  or insincerely ( as they do not get the correct alternative for both states in each of these cases).  
Consider delegating to any other voter $j\neq i$: there exists a profile for each such $j$ such that $\sigma_j(a)=\sigma_j(b)= b$. This is strictly worse than $\sigma_i(a)=a,\sigma_i(b)=a$.

Similarly for the case when $N$ is even, then let $\frac{N-2}{2}$ vote for $a$ and $\frac{N-2}{2}+1$ vote $b$ - it is a strict best response to vote sincerely, as it creates a tie in state $a$ instead of the outcome $B$ for sure. Voting unresponsively  will deliver $B$ for sure in state $a$ and voting semi sincerely (either $B$ in state $a$ or vice versa) or insincerely will do worse. 

Now we show that in the reduced game, the strategy of delegation to an informed independent is weakly dominant for all non expert independents. Since experts know the state of nature with certainty, delegation maximizes the probability of reaching the correct outcome. We now construct profiles to show that delegation is a strict best response for a non expert independent voter $i$  relative to each of the other pure strategies.  $max(n_e+n_A,n_e+n_B) \leq \frac{N}{2}$ is needed to guarantee each nonperfectly informed independent $i$ is pivotal under some profile in each state, otherwise the outcome is determined by partisan voters.  

W.l.o.g let $n_A-n_B=z\geq 0$.  Let $y= n_e+z$. Then consider the profile where partisan and expert voters use their dominant strategies and $y$ of the remaining independent voters excluding voter $i$ use
$\sigma_i(a)= \sigma_i(b)= b,$ while the rest abstain unresponsively.   
This profile is feasible if $n_U\geq y+1,$ which holds by assumption. 
Therefore $i$ is pivotal in state $a$ but not in state $b$ Where $B$ is the outcome). The correct outcome in state $a$ is reached by delegating to any expert. However, if $i$ votes unresponsively for $b$,or $x$ or delegates to any non expert independent or a $B$ partisan the outcome is $B$ in state $a$ as well as in state $b$. 

To show that voting unresponsively for $A$ (or delegating to an $A$ partisan)  is also strictly worse than delegation to an expert, consider the profile where $\tilde y= n_e-(n_A-n_B)$ of the non expert independents (excluding voter $i$) vote unresponsively for $A$ , and the rest abstain. Since $\tilde y\leq  y$ such a profile is feasible by assumption.  Then voter $i$ is pivotal in state $b$ but not in state $a$ (where the outcome is $A$). By delegating to an expert $i$ can ensure the correct outcome in each state. By voting unresponsively for $A$, or delegation to an $A$ partisan, the outcome is $A$ in state $b$.

The case of $n_B-n_A>0$ is exactly symmetric.

Note that under the conditions of the proposition, expert independents decide the election outcome and therefore the outcome is efficient.

 $n_I \geq |n_A-n_B|$ ensures independents are pivotal with positive probability which ensures sincere voting is a weakly dominant strategy for perfectly informed voters ($n_e \geq 1$). 
\end{proof}
\noindent {\bf Proof of Proposition \ref{prop: nodeleg}}
As before, we denote the number of independents who have type $(I,q_i)$, $q_i<1$ as $n_U$.
\begin{prop}
Suppose that $n_e \geq 1$, $n_I \geq |n_A-n_B|$, $max(n_e+n_A,n_e+n_B) \leq \frac{N}{2}$, and that delegation is not allowed. Then: \begin{enumerate}
 \item  The game is not dominance solvable. 
 \item There exist multiple efficient equilibria in asymmetric strategies for nonperfectly informed voters where the outcome is $a$ in state $\omega=a$ and $b$ in state $\omega=b$. Assume $(n_I-n_e)\geq n_I+ |n_A-n_B|+1$.  If $n_e> |n_A-n_B|,$ there exists an efficient equilibrium where all nonperfectly informed independent voters abstain. If $n_e < |n_A-n_B|$ then there exists an inefficient equilibrium where all nonperfectly informed independent voters abstain and the correct outcome is not chosen in at least one state. There also exist inefficient symmetric PSNE where the outcome is $i \in \{a,b\}$ regardless of the state.  
 \end{enumerate}
\end{prop}

\begin{proof}
 
 (1) We use the results of Proposition \ref{prop: dominant} that $A$,$B$ partisan and expert independents have a dominant strategy (a strategy that weakly dominates a set of strategies $S$  also weakly dominates a strict subset of $S$). We now show that in the reduced game without delegation two of the pure strategies (voting unresponsively $a$ or $b$) of a non expert independent $i$ is a unique best response (unique best response) to some profile $\sigma_{-i}$  in the reduced game.

W.l.o.g let $n_A-n_B=z\geq 0$.  Let $y= n_e+z$. Then consider the profile where partisan and expert voters use their dominant strategies and $y$ of the remaining independent voters excluding voter $i$ use
$\sigma_i(a)= \sigma_i(b)= b,$ while the rest abstain unresponsively.   
This profile is feasible if $n_U\geq y+1,$ which holds by assumption. 
Therefore $i$ is pivotal in state $a$ but not in state $b$ where $b$ is the outcome). The correct outcome in state $a$ is reached by voting $a$ unresponsively. However, if $i$ votes unresponsively for $b$ ( $x$) 
the outcome is $b$ (tie) in state $a$ and $b$  in state $b$. Therefore voting unresponsively for $A$ is a unique best response. Moreover voting $x$ unresponsively is a strictly better response than voting $b$ unresponsively. 

To show that voting unresponsively for $B$ is also a unique best response, consider the profile where $\tilde y= n_e-(n_A-n_B)$ of the non expert independents (excluding voter $i$) vote unresponsively for $A$ , and the rest abstain. Since $\tilde y\leq  y$ such a profile is feasible by assumption.  Then voter $i$ is pivotal in state $b$ but not in state $a$ (where the outcome is $a$).  By voting unresponsively for $B$ ,  the outcome is $B$ in state $b$ while voting unresponsively for $A$ yields outcome $A$ in state $b$ and abstaining unresponsively leads to a tie in state $b$. Therefore voting unresponsively for $B$ is a unique best response. Moreover, voting $x$ unresponsively is strictly better than voting $b$ unresponsively. 

Therefore we have shown that voting unresponsively for $A$ or$B$ are unique best response, they are undominated while  abstaining unresponsively is not dominated by any of the other pure strategies. Moreover, since the profiles for showing unique best response for $a$ and $b$ are different, no mixed strategy dominates abstention either. 

The case of $n_B-n_A>0$ is exactly symmetric. 

(2) W.l.o.g assume $n_A-n_B\geq 0$ and $n_e\geq 1$ and let $n_A-n_B$ of the non expert independents vote unresponsively $b$, while the rest abstain. The experts decide the election so it is efficient. This requires $n_I-n_e\geq n_A-n_B$. To generate multiple such equilibria, we need to assume  $n_I-n_e\geq n_A-n_B+1$- permuting the non experts between abstention and voting unresponsively $b$ will not change the outcome as long as numbers remain the same. With sufficiently many non experts there are more equilibria where some of the non experts neutralise each others votes leaving the experts decisive. 
If $n_e> |n_A-n_B|,$ there exists an efficient equilibrium where all nonperfectly informed independent voters abstain. This is easy to see that in this case the experts are  decisive and the correct outcome will be chosen in both  states. If $n_e< |n_A-n_B|,$ the experts are not decisive so if all non experts abstain, the outcome is incorrect in one state.  

Again assume $n_A-n_B\geq 0$. If all non experts vote $a$ unresponsively the outcome is $a$, in both states if there are $ n_U> n_e-(n_A-n_B)$ of non expert independents. This condition ensures that votes for $A$ in state $b$ are $n_A+ n_U> n_B +n_e.$ Similarly when all non experts vote $b$ the outcome is $b$ in state $a$ when $n_B+ n_U> n_A+n_e $

\end{proof}
\end{document}